\documentclass[12pt]{article}
\renewcommand{\baselinestretch}{1.5}
\usepackage{amsfonts,amsmath,amsthm,amssymb}
\usepackage{natbib,color}
\usepackage[pdftex]{graphicx}
\usepackage{subfigure}
\usepackage{enumitem}
\setdescription{leftmargin=0pt}

\usepackage{booktabs}
\allowdisplaybreaks

\definecolor{winered}{rgb}{0.5,0,0}
\usepackage[
pdfstartview=FitH,  bookmarks=true,  bookmarksopen=true,  breaklinks=true,  colorlinks=true,  linkcolor=winered, anchorcolor=my,  citecolor=blue, filecolor=blue,  menucolor=blue, urlcolor=black]{hyperref}
\usepackage[final]{pdfpages}
\makeatletter
\@ifundefined{KV@Gin@artifact}{\define@key{Gin}{artifact}[true]{}}{}
\@ifundefined{KV@Gin@alt}{\define@key{Gin}{alt}{}}{}
\makeatother
\numberwithin{equation}{section}
\newtheorem{theorem}{Theorem}[section]

\newtheorem{proposition}[theorem]{Proposition}

\theoremstyle{definition}

\newtheorem{Examples}{\color{myRed}Examples}
\newtheorem{assumption}[Examples]{Assumption}
\definecolor{my}{rgb}{0.05,0.05,0.5}
\definecolor{myBlue}{rgb}{.1,.1,.5}
\definecolor{myGreen}{rgb}{0,.4,0}
\definecolor{myRed}{rgb}{.25,0.15,.5}

\definecolor{my}{rgb}{0.05,0.05,0.5}

\newcommand{\cond}{\displaystyle \stackrel{d}{\longrightarrow}}

\newcommand{\conas}{\stackrel{a.s.}{\longrightarrow}}
\newcommand{\conp}{\stackrel{p}{\longrightarrow}}
\newcommand{\conpp}{\stackrel{p^{\pi}}{\longrightarrow}}

\renewcommand{\limsup}{\displaystyle \operatornamewithlimits{\lim\sup \ }}
\renewcommand{\liminf}{\displaystyle \operatornamewithlimits{\lim\inf \ }}
\renewcommand{\mathbf}[1]{\textbf{\textit{#1}}}

\usepackage{bbm}
\usepackage{dsfont}

\makeatletter

\newcommand{\IG}{\mathit{\Gamma}}

\newcommand{\E}{\operatorname{E}}
\newcommand{\V}{\operatorname{Var}}

\newcommand{\Rmnum}[1]{\expandafter\@slowromancap\romannumeral #1@}

\newcommand{\blind}{1}
\begin{document}

	\def\spacingset#1{\renewcommand{\baselinestretch}%
		{1.7}\small\normalsize} \spacingset{0.5}
	
	
\if1\blind
{\title{{\vspace{-0.6in}{\bf Robust Permutation Tests in Linear Instrumental Variables Regression}}\thanks{This is an Accepted Manuscript of an article published by Taylor \& Francis in the \emph{Journal of the American Statistical Association} on 3 December 2024, available online: \url{https://www.tandfonline.com/doi/10.1080/01621459.2024.2412363}.}
	}
	\date{October 28, 2024}
	\author{
		Purevdorj Tuvaandorj\footnote{\texttt{tpujee@yorku.ca}}\\York University\\
	}
	\maketitle
} \fi
	
	\if0\blind
	{
			\title{\vspace{-0.6in}\bf Robust Permutation Tests in Linear Instrumental Variables Regression}
		\date{October 28, 2024}
                    \maketitle
	} \fi
	\begin{abstract}
		This paper develops permutation versions of identification-robust tests in linear instrumental variables regression. Unlike the existing randomization and rank-based tests in which independence between the instruments and the error terms is assumed, the permutation Anderson-Rubin (AR), Lagrange Multiplier (LM) and Conditional Likelihood Ratio (CLR) tests are asymptotically similar   
		and robust to conditional heteroskedasticity under standard exclusion restriction i.e. the orthogonality between the instruments and the error terms. 
		Moreover, when the instruments are independent of the structural error term, the permutation AR tests are exact, hence robust to heavy tails. As such, these tests share the strengths of the rank-based tests and the wild bootstrap AR tests.   
		Numerical illustrations corroborate the theoretical results.\par
		\bigskip
		\noindent%
		\vfill 
	\end{abstract}

	\newpage
	\spacingset{1.9} 
	\section{Introduction}
	\label{sec:intro}
	The use of instrumental variables (IVs) is widespread across many disciplines. One of the
	main inferential issues that require taking careful account in the IV regression is to protect against a possible
	weak correlation between instruments and endogenous regressors e.g. treatment variable. To this end, it is desirable to use \emph{identification-robust} tests that are \emph{similar}, at least asymptotically, when the instruments have an arbitrary degree of explanatory power (identification strength) for the endogenous regressor, and have good power when the instruments are informative.\footnote{For a comprehensive treatment of similar tests, see \citet[Chapter 4]{Lehmann-Romano(2005)} and \cite{Linnik(2008)}.\smallskip}\par   
	This paper proposes permutation (randomization) test versions of heteroskedasticity-robust \cite{Anderson-Rubin(1949)}'s AR, \cite{Kleibergen(2002)}'s LM and \cite{Andrews-Guggenberger(2019)}'s CLR tests in IV regression within a super-population framework.\footnote{The LM test is a score test that uses an outer-product-of-the-gradient information matrix estimator, and the CLR test is a heteroskedasticity-robust version of the CLR test proposed by \cite{Moreira(2003)} which is a LR test with  \emph{Neyman structure}, see \citet[Chapter 4.3]{Lehmann-Romano(2005)} for the latter.}  
    We demonstrate the uniform asymptotic similarity of these tests under standard conditions and characterize the conditions under which the permutation AR tests are exact.
	\par 
	Permutation inference is attractive because of its exactness when relevant assumptions hold. It also bears a natural connection to IV methods since 
	IVs provide an exogenous variation independent of the unobserved confounders for the endogenous regressor in the IV regression and the permutation inference typically seeks to exploit independence of two sets of variables by permuting the elements of one variable holding the other fixed to replicate the distribution of a test statistic at hand.\par  
	
	Despite the link, the literature on randomization inference in IV models is scarce. \cite{Imbens-Rosenbaum(2005)} develop exact permutation tests in an IV setting under a finite population framework, where the instruments (or transformations thereof, such as ranks) are permuted while the quantities that do not depend on the IVs are held fixed. Given the advantages of the permutation method in an IV setting as exemplified by \cite{Imbens-Rosenbaum(2005)}, we broaden its scope by proposing the permutation versions of the trinity of aforementioned identification-robust test statistics in a super-population framework. \par 
The contributions of the paper are as follows. First, we propose two permutation AR (PAR) statistics for the IV setting, denoted as \(\mathrm{PAR}_1\) and \(\mathrm{PAR}_2\). These directly extend the work of \cite{Diciccio-Romano(2017)}, who develop permutation tests for regression coefficients in heteroskedastic linear regression, but do not study the IV regression. Specifically, the \(\mathrm{PAR}_1\) statistic permutes the rows of the instrument matrix affecting the endogenous regressor, while the \(\mathrm{PAR}_2\) statistic permutes the null-restricted residuals in the structural equation. We establish the finite sample validity of the PAR tests under an independence assumption similar to that in \cite{Andrews-Marmer(2008)}. We then show their asymptotic validity under the usual exogeneity condition, allowing for conditional heteroskedasticity. As such, our result differs from the finite-population results of \cite{Imbens-Rosenbaum(2005)}.\par 
Second, we propose a permutation LM (PLM) test. Constructing its permutation version is nontrivial because the LM statistic is a nonlinear score statistic. To achieve this, we permute the residuals from the first-stage estimation of the reduced-form equation and generate a permutation endogenous regressor.\footnote{See \cite{Freedman-Lane(1983)} and \cite{Diciccio-Romano(2017)} for residual permutation tests in linear models.} By combining this regressor with the permuted null-restricted residuals in the structural equation, we construct the PLM statistic.\footnote{This approach is distinct from the score statistic construction proposed by \cite{Hemeriketal2020} based on sign-flipping in generalized linear models.}\par 
Third, we propose a permutation CLR (PCLR) test. This test generates the conditional permutation distribution by permuting the null-restricted residuals of the structural equation while keeping the nuisance parameter estimates fixed. Although the PCLR test shares some similarities with the conditional permutation tests proposed by \cite{Rosenbaum1984} and \cite{Hennessyetal2016} in different contexts, such a permutation test has not been considered in the IV literature.\par 
	Fourth, as a main technical contribution, we show the uniform asymptotic similarity of the proposed permutation tests and derive their asymptotic power under local alternatives and strong identification. To the best of our knowledge, these results have not been shown before in the context of permutation/randomization inference. To establish uniform asymptotic similarity, we adapt the asymptotic results from \cite{Andrews-Guggenberger(2019), Andrews-Guggenberger(2019s)} and \cite{Andrews-Cheng-Guggenberger(2020)} to the permutation inference setting.\par 
	
	Fifth, we provide a rigorous analysis of the PAR tests with varying numbers of IVs, identifying the conditions under which these tests maintain asymptotic validity. To substantiate this, we interpret the PAR statistics as double-indexed permutation statistics and make a novel use of the CLT established by \cite{Pham-Mocks-Sroka(1989)}. Such an application appears to be new in the literature.
	
Finally, we compare the robust permutation tests to both the identification-robust rank-based tests of \cite{Andrews-Marmer(2008)} and \cite{Andrews-Soares(2007)}, and the wild bootstrap tests of \cite{Davidson-MacKinnon(2012)} through simulations. In terms of controlling Type I error, we find that the robust permutation tests outperform the rank-based tests under the standard exclusion restriction, and perform on par with or better than the wild bootstrap tests in the heteroskedastic designs considered.\par 
	{\flushleft{\bf Discussion of the related literature.}} The literature on linear IV model is vast. Earlier surveys on this topic are provided by 
	\cite{Stock-Wright-Yogo(2002)}, \cite{Dufour(2003)}, and \cite{Andrews-Stock(2007)}, and recent contributions include \cite{Moreira-Moreira(2019)}, \cite{Young(2020)} and the survey of \cite{Andrews-Stock-Sun(2019)}.\par 
	\cite{Andrews-Marmer(2008)} develop rank-based AR tests under the independence assumption between 
	the instruments and the structural error term. Under their assumption, the $\mathrm{PAR}_1$ test proposed here is exact while the $\mathrm{PAR}_2$ test is so when there are no included non-constant exogenous variables.    
	\cite{Andrews-Soares(2007)} consider a rank-version of the CLR statistic of \cite{Moreira(2003)}. Their test is robust to heavy-tailed errors but not 
	to heteroskedasticity.\par 
\cite{Rosenbaum1984} and \cite{Hennessyetal2016} propose conditional permutation tests in the causal inference context, which is different from our setting. However, the strategies used for the PCLR test validity may prove useful for showing the asymptotic validity of their tests under weak conditions. Also, the PCLR test is quite distinct from the conditional test based on rerandomization by \cite{LiDingRubin2018}, where the conditioning statistic is a permutation quantity as opposed to the sample quantity considered here. \cite{Hemeriketal2020} propose a score statistic in generalized linear models based on sign-flipping. The residual permutation construction of the PLM test provides an alternative method for constructing a permutation score statistic in their framework.\par 
This paper complements the literature on bootstrap inference in the linear IV model, including works such as \cite{Davidson-MacKinnon(2008)}, \cite{Moreira-Porter-Suarez(2009)}, and \cite{Davidson-MacKinnon(2012)}, which document the improved performance of bootstrap inference over asymptotic inference.\par 
	The key to our results is that the identification-robust PAR, PLM, and PCLR statistics are studentized. The use of studentization to obtain an improved test is not new; a well-known example is the higher-order accuracy of the bootstrap-$t$ confidence interval in the one-sample problem \citep[Chapter 15.5]{Lehmann-Romano(2005)}. Other examples include \cite{So-Shin(1999)}, who develop a persistence-robust test for autoregressive models, and \cite{Neuhaus(1993)}, who proposes a permutation test for a two-sample problem with randomly censored data. For further extensions and applications of permutation tests based on studentized statistics, see \cite{Janssen(1997)}, \cite{Neubert-Brunner(2007)}, \cite{Pauly(2011)} and \cite{Chung-Romano(2013), Chung-Romano(2016)}.\par 
	Several recent papers that consider randomization inference in different contexts include 
	\cite{Chernozhukov-Hansen-Jansson(2009)}, 
	\cite{Canay-Romano-Shaikh(2017)}, \cite{Bugni-Canay-Shaikh(2018)}, \cite{Canay-Kamat(2018)}, 
	\cite{Ganong-Jager(2018)} and \cite{Dufour-Flachaire-Khalaf(2019)}.\par
    {\flushleft{\bf Organization of the paper.}}
	The paper is organized as follows. Section \ref{MS} introduces the model and the identification-robust test statistics and develops the permutation tests. Section \ref{sec: Sim} provides simulation results comparing the performance of alternative test procedures. Section \ref{sec: EA} presents an empirical application. We briefly conclude in Section \ref{Conc}. Supplemental Appendix contains
	all the proofs of our theoretical results and additional simulation evidence. We also provide the replication \texttt{R} codes of the empirical applications.
	\section{Robust permutation tests}\label{MS}
	\subsection{Model setup}
	We develop permutation-based tests for a restriction on a scalar structural coefficient 
	\begin{equation}\label{eq: H0}
		H_0:\theta=\theta_0,
	\end{equation}
	in the linear IV regression model:
	\begin{align}
		y
		&=Y\theta+{X}\gamma+u,\label{SE}\\
		Y
		&=W\IG+{X}\mathit{\Psi}+V,\label{RE} 
	\end{align}
	where $y=[y_{1},\dots, y_{n}]^{\prime}$ and 
	$Y=[Y_1,\dots, Y_n]^{\prime}$ are $n$-vectors of endogenous 
	variables, $X=[X_1,\dots, X_n]^{\prime}$ is an $n\times p$ matrix of 
	exogenous covariates whose first column is the $n$-vector of ones $\iota=[1,\dots, 1]^{\prime}$, $W=[W_1,\dots, W_n]^{\prime}$ is an $n\times k\ (k\geq 1)$ 
	matrix of IVs that does not include $\iota$ and the variables in $X$, $u=[u_1, \dots, u_n]^{\prime}$ is an $n$-vector of structural error terms, $V=[V_1,\dots, V_n]^{\prime}$ is an $n$-vector of reduced-form error terms; $\gamma$ is a $p$-parameter vector, and $\IG$ and $\mathit{\Psi}$ are $k$- and $p$-vectors of reduced-form coefficients, respectively.\par 
	Let $Z=[Z_1,\dots, Z_n]^{\prime}\equiv M_{X}W$, where $M_A\equiv I_n-P_A\equiv I_n-A(A^{\prime}A)^{-1}A^{\prime}$ for a matrix $A$ of full column rank. 
For later use, it will be convenient to express the 
	equation \eqref{RE} as 
	\begin{align}\label{IVReg}
		Y
		&=Z\IG+{X}\xi+V, 
	\end{align}
	where $\xi\equiv(X^{\prime}X)^{-1}X^{\prime}W\mathit{\Gamma}+\mathit{\Psi}$. 
We also define $\tilde{u}_i(\theta)\equiv y_i-Y_i^{\prime}\theta-X_i^{\prime}(X^{\prime}X)^{-1}X^{\prime}(y-Y\theta)$, $\tilde{u}(\theta)\equiv [\tilde{u}_1(\theta),\dots, \tilde{u}_n(\theta)]^{\prime}$, $\tilde{Y}=[\tilde{Y}_1,\dots, \tilde{Y}_n]^{\prime}\equiv M_{X} Y$, and 
	\begin{align}
		\hat{m}
		&\equiv n^{-1}Z^{\prime}\tilde{u}(\theta_0),\label{eq: mfun}\\
		\hat{\Sigma}
		&\equiv n^{-1}\sum_{i=1}^nZ_iZ_i^{\prime}\tilde{u}_{i}(\theta_0)^2,\label{eq: Sighat}\\ 
		\hat{C}
		&\equiv n^{-1}\sum_{i=1}^nZ_iZ_i^{\prime}\tilde{Y}_{i}\tilde{u}_i(\theta_0),\label{eq: Chat}\\
		\hat{G}
		&\equiv n^{-1}Z^{\prime}Y,\quad \hat{J}
		\equiv\hat{G}-
		\hat{C}\hat{\Sigma}^{-1}\hat{m}.\label{eq: Jacobs} 
	\end{align} 
The first heteroskedasticity and identification-robust test statistic is the AR statistic defined as 
	 \begin{equation*}
		\text{AR}=\text{AR}(W, X, \tilde{u}(\theta_0))\equiv n\,\hat{m}^{\prime}\hat{\Sigma}^{-1}\hat{m}.  
	\end{equation*}
	The null asymptotic distribution of the AR statistic is chi-square regardless of the values of $\IG$ because $n^{1/2}\hat{m}$ is asymptotically normal under \eqref{eq: H0}. Another asymptotically pivotal test statistic
is the following heteroskedasticity-robust version of \cite{Kleibergen(2002)}'s LM statistic
	\begin{equation}
		\text{LM}=\mathrm{LM}(Z, \tilde{u}(\theta_0), Y)\equiv n\,\hat{m}^{\prime}\hat{\Sigma}^{-1/2}P_{\hat{\Sigma}^{-1/2}\hat{J}}\hat{\Sigma}^{-1/2}\hat{m}. 
	\end{equation}
	As defined, $\hat{J}$ in \eqref{eq: Jacobs} is a vector of residuals from the regression of the sample Jacobian $\hat{G}$ on the sample moment function in \eqref{eq: mfun}. When $\IG$ has full rank i.e. the identifying power of $W$ is strong, $\hat{J}$ converges to a matrix of full rank under mild conditions, and the LM statistic is asymptotically chi-squared under \eqref{eq: H0}. However, even when $\IG$ is ``small'' so that the noise and signal parts of $Y$ are of a similar magnitude or the former dominates the latter,  
	$\hat{J}$ and $\hat{m}$ are asymptotically independent after suitable normalizations, and consequently, the LM statistic is still chi-squared in the limit.\par 
	
	The CLR-type test statistic considered next has an approximate Neyman structure, such that its conditional asymptotic distribution, given a sufficient statistic for the vector of nuisance parameters $\IG$ is independent of $\IG$. Let 
	\begin{align}
		\hat{\mathcal{S}}
		&\equiv\hat{\Sigma}^{-1/2}n^{1/2}\hat{m},
		\label{SSuff}\\
		\hat{\mathcal{T}}
		&\equiv\hat{\Sigma}^{-1/2}n^{1/2}\hat{J}{\sqrt{(\theta_0, 1){\Omega}^{\epsilon}(\hat{\Sigma}, \hat{\mathcal{V}})^{-1}(\theta_0, 1)^{\prime}}},\label{TSuff0}
	\end{align}
	where $\Omega^\epsilon(\hat{\Sigma}, \hat{\mathcal{V}})$ is an eigenvalue-adjusted version of the $2\times 2$ matrix defined as 
	\begin{align}
		{\Omega}(\hat{\Sigma}, \hat{\mathcal{V}})
		&\equiv	k^{-1}\begin{bmatrix}
		\mathrm{tr}({K}_{11}(\hat{\mathcal{V}})^{\prime}\hat{\Sigma}^{-1})& 	\mathrm{tr}({K}_{12}(\hat{\mathcal{V}})^{\prime}\hat{\Sigma}^{-1})\\
		\mathrm{tr}({K}_{21}(\hat{\mathcal{V}})^{\prime}\hat{\Sigma}^{-1})& 	\mathrm{tr}({K}_{22}(\hat{\mathcal{V}})^{\prime}\hat{\Sigma}^{-1})
		\end{bmatrix},\label{OH}
	\end{align}
	and the $k\times k$ matrix ${K}_{ij}(\hat{\mathcal{V}})$ is the $(i, j)$ submatrix of ${K}(\tilde{\mathcal{V}})$ given by  
	\begin{align}
		{K}(\hat{\mathcal{V}})
		&\equiv (B(\theta_0)^{\prime}\otimes I_{k})\,\hat{\mathcal{V}}\,(B(\theta_0)\otimes I_{k})=\begin{bmatrix}
			{K}_{11}(\hat{\mathcal{V}})& {K}_{12}(\hat{\mathcal{V}})\\
		{K}_{21}(\hat{\mathcal{V}})& {K}_{22}(\hat{\mathcal{V}})
		\end{bmatrix},\quad 
		B(\theta)\equiv
		\begin{bmatrix}
			1& 0\\
			-\theta& -1
		\end{bmatrix}.
		\label{KH} \text{\footnotemark}
	\end{align}
With an appropriate choice of $\hat{\mathcal{V}}$ specified below, the transformation $\hat{\mathcal{T}}$ in \eqref{TSuff0} provides an asymptotically sufficient statistic for $\IG$ that is independent of the statistic $\hat{\mathcal{S}}$ in \eqref{SSuff}. 
\cite{Andrews-Guggenberger(2019), Andrews-Guggenberger(2019s)} propose two different versions of $\hat{\mathcal{V}}$, and for the linear IV model, they recommend the following choice of $\hat{\mathcal{V}}$:
		\begin{align}
		\hat{\mathcal{V}}
		&\equiv n^{-1}\sum_{i=1}^n\left[(\varepsilon_i(\theta_0)-\hat{\varepsilon}_{in}(\theta_0))
		(\varepsilon_i(\theta_0)-\hat{\varepsilon}_{in}(\theta_0))^{\prime}\right]\otimes (Z_iZ_i^{\prime}),\label{hvb}\\
		\varepsilon_i(\theta_0)
		&\equiv\left[\tilde{u}_i(\theta_0), -\tilde{Y}_i\right]^{\prime},\quad 
		\hat{\varepsilon}_{in}(\theta_0)\equiv \varepsilon(\theta)^{\prime}Z(Z^{\prime}Z)^{-1}Z_i,\quad \varepsilon(\theta)\equiv[\varepsilon_1(\theta),\dots, \varepsilon_n(\theta)]^{\prime}.
	\end{align}
	The CLR-type statistic is then defined as 
	\begin{align}
		\text{CLR}&=\text{CLR}(\hat{\mathcal{S}},\hat{\mathcal{T}})\equiv\hat{\mathcal{S}}^{\prime}{\hat{\mathcal{S}}}-\lambda_{\mathrm{min}}
		\left[(\hat{\mathcal{S}}, \hat{\mathcal{T}})^{\prime}(\hat{\mathcal{S}}, \hat{\mathcal{T}})\right],\label{CLR}
	\end{align}
	where $\lambda_{\min}(\cdot)$ denotes the smallest eigenvalue of a matrix. In Supplemental Appendix, we consider another CLR-type statistic of \cite{Andrews-Guggenberger(2019)} based on a different specification for $\hat{\mathcal{V}}$. The $\text{CLR}$ in \eqref{CLR} is denoted as $\text{QLR}_P$ in \citet[p.10]{Andrews-Guggenberger(2019s)},  and it is shown to be asymptotically equivalent to \cite{Moreira(2003)}'s CLR statistic 
	in the homoskedastic linear IV regression with fixed instruments and multiple endogenous variables under all strengths of identification. 
	\footnotetext{\label{evaladj}The eigenvalue-adjustment of \cite{Andrews-Guggenberger(2019s)} is as follows. 
		Let $A$ be a nonzero positive semi-definite matrix of dimension $p\times p$ that has a spectral decomposition $A=N\Delta N^{\prime}$, where $\Delta=\mathrm{diag}(\lambda_1,\dots,\lambda_p)$, 
		$\lambda_1\geq\dots\geq\lambda_{p}\geq 0,$ is the diagonal matrix that consists of the eigenvalues of $A$, and 
		$N$ is an orthogonal matrix of the corresponding eigenvectors. Given a constant $\epsilon>0$, 
		the eigenvalue adjusted matrix is defined as $A^{\epsilon}\equiv N \Delta^{\epsilon}N^{\prime}$, where 
		$\Delta^{\epsilon}\equiv\mathrm{diag}(\max\{\lambda_1,\lambda_{1}\epsilon\},\dots,  \max\{\lambda_p,\lambda_{1}\epsilon\})$. They recommend $\epsilon=0.01$. We refer to 
		\cite{Andrews-Guggenberger(2019)} for further properties of the eigenvalue-adjustment procedure.}
	\subsection{Main results}\label{PT}
	We begin by recalling the basic notion of randomization tests from Chapter 15.2 of \cite{Lehmann-Romano(2005)}. 
	Let $\mathbb{G}_n$ be the set of all permutations $\pi=[\pi(1), \dots, \pi(n)]$ of $\{1,\dots, n\}$, and denote a permutation version of a generic test statistic $R$ by $\text{PR}=R^\pi$. 
The distribution function corresponding to permutations in $\mathbb{G}_n$ is  
	\begin{equation}\label{eq: Perm dfun}
		\hat{F}_{n}^{\text{PR}}(x)\equiv\frac{1}{n!}\sum_{\pi\in\mathbb{G}_n}1(R^{\pi}\leq x),
	\end{equation}
where $1(\cdot)$ is the indicator function. The computation of \eqref{eq: Perm dfun} and related quantities (at a particular value of $x$) for permutation tests entails computing $n!$ test statistics which is impractical even for $n$ relatively small. So instead, a stochastic approximation based on uniform random draws from $\mathbb{G}_n$ is often used. To this end, let us fix $N\in \{1,\dots,|n!|\}$ and consider a subset of permutations $\mathbb{G}_n'=\{\pi_1,\dots, \pi_N\}\subseteq \mathbb{G}_n$, where $\pi_1$ is the identity permutation, and $\pi_2,\dots, \pi_N$ are i.i.d. uniformly distributed on $\mathbb{G}_n$. 
	The permutation distribution function corresponding to $\mathbb{G}_n'$ is  
	\begin{equation}\label{eq: Perm dfun Gn'}
		\tilde{F}_{N}^{\text{PR}}(x)\equiv{N}^{-1}\sum_{\pi\in\mathbb{G}_n'}1(R^{\pi}\leq x).
	\end{equation}
	Let $R^{\pi}_{(1)}\leq \dots\leq R^{\pi}_{(N)}$ be the order statistics of  
	$\{R^\pi: \pi\in\mathbb{G}_n'\}$, and for a nominal significance level $\alpha\in(0,1)$, define 
	\begin{align}
		r&\equiv N-I(N\alpha),\label{rdef}\\
		N^{+}&\equiv \vert\{j=1,\dots, N: R_{(j)}^{{\pi}}>R^\pi_{(r)}\}\vert ,\quad N^{0}\equiv \vert\{j=1,\dots, N: R_{(j)}^{{\pi}}=R^\pi_{(r)}\}\vert ,\notag\\
		a^{\text{PR}}&\equiv (N\alpha-N^{+})/N^0,\label{eq: test a def}
	\end{align}
	where $I(\cdot) $ denotes the integer part of a number. Clearly, $N^{+}$ and $N^{0}$ are the number of values $R_{(j)}^{{\pi}}=\mathrm{PR}_{(j)}, j = 1,\dots, N$, that are greater than 
	$R_{(r)}^\pi$ and equal to $R_{(r)}^\pi$, respectively. A randomization test is then defined as 
	\begin{equation}
		\phi_n^{\text{PR}}
		\equiv 
		\begin{cases}
			1 &\text{if}\  R>R^\pi_{(r)},\\
			a^{\text{PR}}&\text{if}\  R=R^\pi_{(r)},\\
			0 &\text{if}\  R<R^\pi_{(r)}.
		\end{cases}
	\end{equation}
    The finite sample results presented below hold for any $N$ but 
    we will assume that $N\to\infty$ as $n\to \infty$ for the asymptotic results.  
	We next describe the robust permutation test statistics. 
	Let $W_\pi=[W_{\pi(1)},\dots, W_{\pi(n)}]^{\prime}$ and 
	$\tilde{u}_{\pi}(\theta)=[\tilde{u}_{\pi(1)}(\theta),\dots, \tilde{u}_{\pi(n)}(\theta)]^{\prime}$  
	be an instrument matrix and a residual vector obtained by permuting the rows of $W$ and the elements of $\tilde{u}(\theta)$, respectively, for a permutation $\pi\in\mathbb{G}_n'$. Set 
	\begin{equation*}
		\tilde{\Sigma}_u\equiv \mathrm{diag}(\tilde{u}_1(\theta_0)^2,\dots, 
		\tilde{u}_n(\theta_0)^2),\ \hat{m}^\pi\equiv n^{-1}Z^{\prime}\tilde{u}_\pi(\theta_0)
		,\   
		\hat{\Sigma}^\pi\equiv n^{-1}\sum_{i=1}^nZ_iZ_i^{\prime}\tilde{u}_{\pi(i)}(\theta_0)^2.
	\end{equation*}  
	We consider two heteroskedasticity-robust permutation AR statistics defined as: 
	\begin{align}
		\text{PAR}_1
		&\equiv \text{AR}(W_\pi, X, \tilde{u}(\theta_0))=\tilde{u}(\theta_0)^{\prime}M_XW_\pi
		(W_\pi^{\prime}M_X\tilde{\Sigma}_u M_XW_\pi)^{-1}W_\pi^{\prime}M_X\tilde{u}(\theta_0),\\
		\text{PAR}_2
		&\equiv \text{AR}(W, X, \tilde{u}_{\pi}(\theta_0))
		=n\,\hat{m}^{\pi\prime}\hat{\Sigma}^{\pi-1}\hat{m}^{\pi}.\label{ARpi}
	\end{align}
	The $\text{PAR}_1$ statistic is based on the permutation of the rows of the instrument matrix $W$, and the $\text{PAR}_2$ statistic uses the permutation of the null-restricted residuals $\tilde{u}(\theta_0)$. 
	The finite sample validity of these tests are shown under the following condition. 
	\begin{assumption}\label{exact}
		$\{(W_i^{\prime}, X_i^{\prime}, u_i)^{\prime}\}_{i=1}^n$ are i.i.d., and 
		$W_i$ and $[X_i^{\prime}, u_i]^{\prime}$ are independently distributed. 
	\end{assumption}
	\cite{Andrews-Marmer(2008)} develop 
	exact rank-based AR tests under assumptions similar to Assumption \ref{exact}.  
	When $W$ is independent of $X$ and $u$, $W^{\prime}M_X\tilde{u}(\theta_0)$ and 
	$W_\pi^{\prime}M_X\tilde{u}(\theta_0)$ have the same distribution under $H_0:\theta=\theta_0$, so 
	the $\text{PAR}_1$ test is exact. On the other hand, the $\text{PAR}_2$ test is not, in general, exact because $W^{\prime}M_X\tilde{u}(\theta_0)$ and 
	$W^{\prime}M_X\tilde{u}_\pi(\theta_0)$ may not have identical distributions.  
	However, when $X=\iota$ in Assumption \ref{exact}, 
	$Z$ and $u_\pi$ are independent,  
	$Z^{\prime}\tilde{u}(\theta_0)$ and 
	$Z^{\prime}\tilde{u}_\pi(\theta_0)$ are identically distributed, and consequently, the $\text{PAR}_2$ test is exact. 
	The following result summarizes the finite sample validity of the PAR tests. 
	\begin{proposition}[Finite sample validity]\label{ExactAR}
		Under Assumption \ref{exact} and $H_0:\theta=\theta_0$, $\E[\phi_n^{\mathrm{PAR}_1}]=\alpha$ for $\alpha\in(0,1)$. If $X=\iota$ in Assumption \ref{exact}, 
		then $\E[\phi_n^{\mathrm{PAR}_2}]=\alpha$.
	\end{proposition}
	Even if $[X_i^{\prime}, u_i]^{\prime}$ and $W_i$ are not distributed independently but satisfy the orthogonality condition $\E[(W_i^{\prime}, X_i^{\prime})^{\prime}u_i]=0$, we show that the permutation AR tests are asymptotically similar and heteroskedasticity-robust.\par
	The argument for permutation LM statistic is more involved because it is difficult to 
	construct an estimator of the reduced-form coefficients $\mathit{\Gamma}$ whose (asymptotic) distribution remains invariant to a permutation of the data. 
	The OLS estimators of the reduced-form coefficients and the corresponding residuals are 
	\begin{equation*}
		\mathit{\hat{\Gamma}}=(Z^{\prime}Z)^{-1}Z^{\prime}Y,\quad\hat{\xi}=(X^{\prime} X)^{-1}X^{\prime}Y,\quad  
		\hat{V}=Y-Z\mathit{\hat{\Gamma}}-X\hat{\xi}=[\hat{V}_1,\dots, \hat{V}_n]^{\prime}.
	\end{equation*}
	Now permute the residuals 
	$\hat{V}_{\pi}=[\hat{V}_{\pi(1)},\dots, \hat{V}_{\pi(n)}]^{\prime}$, and let
	\begin{equation*}
		Y^{\pi}=[Y_{1}^{\pi},\dots, Y_{n}^{\pi}]^{\prime}\equiv Z\mathit{\hat{\Gamma}}+X\hat{\xi}+\hat{V}_\pi.
	\end{equation*}  
The idea of permuting the residuals to obtain asymptotically valid randomization tests appears in \cite{Freedman-Lane(1983)} and \cite{Diciccio-Romano(2017)}.\par
The Jacobian estimator used in the permutation LM statistic is 
	\begin{align}
		\hat{J}^{\pi}
		&\equiv n^{-1}Z^{\prime}Y^{\pi}-\hat{C}^\pi\hat{\Sigma}^{\pi-1}\hat{m}^{\pi},\label{probJ}\\
		\hat{C}^\pi
		&\equiv n^{-1}\sum_{i=1}^nZ_iZ_i^{\prime}\hat{V}_{\pi(i)}\tilde{u}_{\pi(i)}(\theta).
	\end{align}
	It is not difficult to see that in strongly identified models where $\IG$ is a fixed vector of full rank, $\hat{J}^{\pi}-n^{-1}Z^{\prime}Z\IG\conpp 0$ in probability, where $\conpp$ denotes the convergence in probability with respect to the probability measure of $\pi$ uniformly distributed over $\mathbb{G}_n$ conditional on the data.  
	The permutation LM statistic is then defined as 
	\begin{equation}\label{LM1}
		\mathrm{PLM}\equiv \mathrm{LM}(Z, \tilde{u}_\pi(\theta_0), Y^{\pi})=n\,\hat{m}^{\pi\prime}\hat{\Sigma}^{\pi-1/2}P_{\hat{\Sigma}^{\pi-1/2}\hat{J}^\pi}\hat{\Sigma}^{\pi-1/2}\hat{m}^\pi. 
	\end{equation}
	Finally, we define the permutation CLR statistic as follows: 
	\begin{align}
		\mathrm{PCLR}
		&\equiv \mathrm{CLR}(\hat{\mathcal{S}}^\pi, \hat{\mathcal{T}})=\hat{\mathcal{S}}^{\pi\prime}\hat{\mathcal{S}}^{\pi}-\lambda_{\mathrm{min}}
		\left[(\hat{\mathcal{S}}^{\pi}, \hat{\mathcal{T}})^{\prime}(\hat{\mathcal{S}}^{\pi}, \hat{\mathcal{T}})\right],\label{PCLRb}
	\end{align}
	where $\hat{\mathcal{T}}$ is as defined in \eqref{hvb}, and 
	\begin{align}
		\hat{\mathcal{S}}^{\pi}
		&\equiv n^{1/2}\hat{\Sigma}^{\pi-1/2}\hat{m}^\pi=\left(\sum_{i=1}^nZ_iZ_i^{\prime}\tilde{u}_{\pi(i)}(\theta_0)^2\right)^{-1/2}Z^{\prime}\tilde{u}_\pi(\theta_0).
		\label{PS}
	\end{align}
	Analogously to \eqref{eq: Perm dfun} and \eqref{eq: Perm dfun Gn'}, define the permutation distributions of the statistic  $\mathrm{PCLR}=\mathrm{CLR}(\hat{\mathcal{S}}^\pi, \hat{\mathcal{T}})$ given $\hat{\mathcal{T}}$, as $\hat{F}_n^{\mathrm{PCLR}}(x, \hat{\mathcal{T}})
\equiv(n!)^{-1}\sum_{\pi\in\mathbb{G}_n}1(\mathrm{CLR}(\hat{\mathcal{S}}^\pi, \hat{\mathcal{T}})\leq x)$ and 
	\begin{equation}
\tilde{F}_N^{\mathrm{PCLR}}(x, \hat{\mathcal{T}})
	\equiv N^{-1}\sum_{\pi\in\mathbb{G}_n'}1(\mathrm{CLR}(\hat{\mathcal{S}}^\pi, \hat{\mathcal{T}})\leq x).
\end{equation}
Let $\mathrm{PCLR}_{(r)}(\hat{\mathcal{T}})$ be the $r$-th order statistic of $\{\mathrm{CLR}(\hat{\mathcal{S}}^\pi, \hat{\mathcal{T}}): \pi\in\mathbb{G}_n'\}$ for $r$ defined in (\ref{rdef}), and $a^{\mathrm{PCLR}}$ be as in \eqref{eq: test a def}. The nominal level $\alpha$ PCLR test $\phi_n^{\mathrm{PCLR}}$ is 1 if $\mathrm{CLR}(\hat{\mathcal{S}}, \hat{\mathcal{T}})>\mathrm{PCLR}_{(r)}(\hat{\mathcal{T}}),$ $a^{\mathrm{PCLR}}$ if $\mathrm{CLR}(\hat{\mathcal{S}}, \hat{\mathcal{T}})=\mathrm{PCLR}_{(r)}(\hat{\mathcal{T}})$ and $0$ otherwise. \par   
	Let $P$ denote the distribution of the vector $(W_i^{\prime}, X_i^{\prime}, u_i, V_i)^{\prime}$ and we shall index the relevant quantities by $P$ in the sequel. Define 
	\begin{equation}\label{eSig}
		e_i\equiv[u_i, V_i^{\prime}]^{\prime},\quad 
		\Sigma^{e}_P
		\equiv \V_P[e_i]
		=
		\begin{bmatrix}
			\sigma^2_P& \Sigma^{uV}_P\\
			\Sigma^{Vu}_P& \Sigma^{V}_P
		\end{bmatrix},
	\end{equation}
	$Z_i^{*}\equiv W_i-\E_P[W_iX_i^{\prime}](\E_P[X_iX_i^{\prime}])^{-1}X_i$, and 
	\begin{equation}\label{Vb}
		\mathcal{V}_{P}\equiv\E_P\left[
		\begin{pmatrix}
			u_i^2& -u_iV_i\\
			-u_iV_i& V_i^2
		\end{pmatrix}
		\otimes 
		(Z_i^{*}Z_i^{*\prime})
		\right]. 
	\end{equation}
	We maintain the following assumptions to develop the asymptotic results. 
	\begin{assumption}\label{A1}
		\begin{enumerate}[label=(\alph*)]
			For some constants $\delta, \delta_1>0$ and $M_0<\infty$, 
			\item
			$\{(W_i^{\prime}, X_i^{\prime}, u_i, V_i)^{\prime}\}_{i=1}^n$ are i.i.d. with distribution $P$,\label{sa}
			\item $\E_P[u_i(W_i^{\prime}, X_i^{\prime})]=0$,\label{or}
			\item $\E_P[\Vert(W_i^{\prime}, X_i^{\prime}, u_i)^{\prime}\Vert^{4+\delta}]<M_0$,\label{mo}
			\item $\lambda_{\min}(A)\geq \delta_1$ for $A\in\big\{\E_P[(W_i^{\prime}, X_i^{\prime})^{\prime}(W_i^{\prime}, X_i^{\prime})], \E_P[{Z}_i^{*}{Z}_i^{*\prime}],  
			\Sigma_P, \sigma_P^2\}$,\label{in}
			\item $\E_P[V_i(W_i^{\prime}, X_i^{\prime})]=0$,\label{or2}
			\item $\E_P[\vert V_i\vert^{4+\delta}]<M_0$,\label{m2}
			\item $\lambda_{\min}(\Sigma_P^V-\Sigma_P^{Vu}(\sigma_P^2)^{-1}\Sigma_P^{uV})\geq \delta_1$.\label{in3}
		\end{enumerate}
	\end{assumption}
	Next we define the parameter space for the permutation tests. 
	\begin{align}
		\mathcal{P}_0^{\mathrm{PAR}_1}
		&=\mathcal{P}_0^{\mathrm{PAR}_2}\equiv\{P: \text{Assumption \ref{A1}\ref{sa}-\ref{in} hold}\},\label{PSPAR}\\
		\mathcal{P}_0^{\mathrm{PLM}}
		&\equiv \{P: \text{Assumption \ref{A1} holds}\},\label{PSPLM}\\
		\mathcal{P}^{\mathrm{PCLR}}_0
		&\equiv\{P: \text{Assumption \ref{A1}\ref{sa}-\ref{m2} hold}\}.
	\end{align}
	Assumption \ref{A1} does not impose virtually any restriction on the degree of association between $Y$ and $W$ (after controlling for the effect of the exogenous covariates), thus allowing for an arbitrary identification strength on the part of the instruments.\par  
	The distribution $P$ in Assumption \ref{A1}\ref{sa} is allowed to vary with the sample size i.e. $P=P_n$. For simplicity, we suppress the dependence on $n$. 
	Assumption \ref{A1}\ref{mo} and \ref{m2} impose finite $4+\delta$ moment on the error terms and the exogenous variables, and are slightly stronger than the cross moment restrictions used in \cite{Andrews-Guggenberger(2019)}. For pointwise asymptotic results or when $P$ does not vary with $n$, weaker restrictions would be sufficient.\par  
	The independence assumption between the instruments and the error terms, which is typically made in the randomization test literature (see, for example, \cite{Imbens-Rosenbaum(2005)}), is not maintained here; the instruments are only assumed to satisfy the standard exogeneity conditions in Assumption \ref{A1}\ref{or} and \ref{or2}. If independence assumptions are maintained, as shown in Proposition \ref{ExactAR}, the 
	PAR tests are exact and do not require the finite $4+\delta$ moment assumptions.\par  
	Assumption \ref{A1}\ref{in} and \ref{in3} require that the second moment matrices of the instruments and exogenous covariates, and covariance matrices are (uniformly) nonsingular, and are similar to the assumptions employed by \cite{Andrews-Guggenberger(2017), Andrews-Guggenberger(2019)} and \cite{Andrews-Cheng-Guggenberger(2020)} when developing identification-robust (but not singularity-robust) AR, LM and CLR tests. These conditions can be restrictive but can potentially be relaxed as in  
		 \cite{Andrews-Guggenberger(2019)} (see also \cite{Dufour-Taamouti(2007)}) using generalized inverses of the relevant matrices. However, doing so would entail a considerable complication in the proof, and so, for simplicity, we maintain Assumption \ref{A1}\ref{in} and \ref{in3}.\par 
	The condition $\lambda_{\min}(\Sigma_P^V-\Sigma_P^{Vu}(\sigma_P^2)^{-1}\Sigma_P^{uV})\geq \delta_1$  is needed for the asymptotic similarity of the PLM test.\par   
	The main result of this paper is given in the following theorem which establishes the uniform asymptotic validity of the robust permutation tests. 
	\begin{theorem}[Asymptotic similarity]\label{UAV}
	Suppose that $H_0:\theta=\theta_0$ holds, and $N\to\infty$ and $n\to\infty$. Then, for $\alpha\in (0,1)$ and the permutation statistics 
		\begin{equation*}
			\mathrm{PR}\in\{\mathrm{PAR}_1, \mathrm{PAR}_2, \mathrm{PLM}, \mathrm{PCLR}\}
		\end{equation*}
		with the corresponding 
		parameter spaces $\mathcal{P}_0\in\{\mathcal{P}^{\mathrm{PAR}_1}_0, \mathcal{P}^{\mathrm{PAR}_2}_0, \mathcal{P}^{\mathrm{PLM}}_0, \mathcal{P}^{\mathrm{PCLR}}_0\}$, respectively,  
		\begin{equation*}
			\limsup_{n\rightarrow \infty}\sup_{P\in\mathcal{P}_0} \E_{P}[\phi_n^{\mathrm{PR}}]
			=\liminf_{n\rightarrow \infty}\inf_{P\in\mathcal{P}_0} \E_{P}[\phi_n^{\mathrm{PR}}]
			=\alpha.
		\end{equation*}
	\end{theorem}
	Theorem \ref{UAV} shows that the permutation tests are asymptotically similar (hence asymptotically of correct size) i.e. the asymptotic null rejection probability is equal to the nominal significance level over a large class of data generating processes.
	Thus, the advantage of the $\mathrm{PAR}_1$ and $\mathrm{PAR}_2$ tests relative to the existing tests is that they 
	inherit the desirable properties of both the exact and asymptotic AR tests (including other simulation-based tests that are 
	asymptotically valid); 
	the PAR tests 
	are exact under the same assumptions as the rank-based tests are, and asymptotically pivotal and heteroskedasticity-robust under the assumptions used to show the validity of the asymptotic tests. \par 
	The main technical tool for the asymptotics of the permutation statistics are 
	the combinatorial central limit theorem based on a Lindeberg type condition \citep{Motoo(1956)} for the sums
	\begin{equation}\label{eq: key quants}
		n^{-1/2}\sum_{i=1}^nZ_i\tilde{u}_{\pi(i)}(\theta_0),\quad n^{-1/2}\sum_{i=1}^nZ_i\hat{V}_{\pi(i)}.
	\end{equation}
    The fact that the intercept term $\iota$ is included in $X$ in \eqref{SE}-\eqref{RE} implies that the two quantities in \eqref{eq: key quants}
    have mean zero with respect to the distribution of $\pi$, which, in turn, plays an important role for the asymptotic validity of the permutation tests (see also Remark 3.1 of 
    \cite{Diciccio-Romano(2017)}).\par 
	The proof of asymptotic similarity of the permutation tests uses the generic method of \cite{Andrews-Cheng-Guggenberger(2020)} for establishing the asymptotic size of tests.\par 
	Next, we shall derive the asymptotic distribution of the permutation statistics under 
	the local alternatives $\theta_n=\theta_0+h_\theta n^{-1/2}$ with $h_\theta \in\mathbb{R}$ fixed, assuming 
	strong identification. Let $\chi_l^2(\eta^2)$ denote noncentral chi-square random variable with degrees of freedom $l$ and noncentrality parameter $\eta^2$. 
\begin{proposition}[Asymptotic local power under strong identification]\label{LP}
	Let Assumption \ref{A1} hold, and $h_\theta \in\mathbb{R}$ be fixed. Assume that the model is strongly identified in the sense that 
$n^{1/2}\Vert\mathit{\Gamma}\Vert\to\infty$ as $n\to\infty$, and for $G_P\equiv\E_P[{Z}_i^{*}{Z}_i^{*\prime}]\mathit{\Gamma}$, $\eta^2\equiv \lim_{n\to\infty}G_P'\Sigma_P^{-1}G_Ph_\theta^2$ exists. Then, under $H_1:\theta_n=\theta_0+h_\theta n^{-1/2}$,
	$\mathrm{AR}\cond \chi^2_k(\eta^2)$, $\mathrm{LM}\cond \chi^2_1(\eta^2)$  and $\mathrm{CLR}\cond \chi^2_1(\eta^2)$ as $n\to\infty$, and 
	\begin{align}
		&\sup_{x\in\mathbb{R}}\vert \tilde{F}_N^{\mathrm{PAR}_1}(x)-P[\chi^2_k\leq x]\vert \conp 0,\label{PARDist}\\
		&\sup_{x\in\mathbb{R}}\vert \tilde{F}_N^{\mathrm{PAR}_2}(x)-P[\chi^2_k\leq x]\vert \conp 0,\label{PARDist2}\\
		&\sup_{x\in\mathbb{R}}\vert \tilde{F}_N^{\mathrm{PLM}}(x)-P[\chi^2_1\leq x]\vert \conp 0,\label{PLMDist}\\
		&\sup_{x\in\mathbb{R}}\vert \tilde{F}_N^{\mathrm{PCLR}}(x, \hat{\mathcal{T}})-P[\chi^2_1\leq x]\vert\conp 0,\label{PCLRDist}
	\end{align}
as $N\to\infty$ and $n\to\infty$.
\end{proposition}
It follows that the asymptotic local power of the PAR tests is equal to that of the asymptotic AR test given by 
$P[\chi_k^2(\eta^2)>r_{1-\alpha}(\chi^2_k)]$, where $r_{1-\alpha}(\chi^2_k)$ is the $1-\alpha$ quantile of $\chi^2_k$ random variable. Moreover, under strong identification, the PLM and PCLR tests, and the asymptotic LM and CLR tests have the same asymptotic local power equal to $P[\chi_1^2(\eta^2)>r_{1-\alpha}(\chi^2_1)]$.
	\subsection{PAR tests with a diverging number of IVs}\label{subsec: manyIV}
This section considers an extension of the PAR tests to cases where the number of IVs, $k$, may grow with the sample size $n$. This extension comprises two main parts. First, the derivation of the asymptotic distribution of the PAR statistics requires a proof strategy different from the fixed-$k$ case. This is accomplished by 
interpreting a dominant term in the PAR statistic (after centering and scaling) as a double-indexed permutation statistic. 
Second, we derive the asymptotic distribution of the AR statistic under the condition that $k^3/n\to 0$ as $n\to\infty$, 
which is the same rate used by \cite{Andrews-Stock(2007)} who establish the asymptotic validity of the AR, LM and CLR tests in homoskedastic linear IV models.\par 

For simplicity, we assume that $X=\iota$, although the proof can be extended to incorporate a fixed number of exogenous covariates in a straightforward manner. In the present case, the $\mathrm{PAR}_1$ and $\mathrm{PAR}_2$ tests are (distributionally) equivalent because 
$W_\pi' M_{\iota}u(\theta_0)=\sum_{i=1}^n(W_{\pi(i)}-\bar{W})u_i(\theta_0)=\sum_{i=1}^n(W_{i}-\bar{W})u_{\pi^{-1}(i)}(\theta_0)=Z'u_{\pi^{-1}}(\theta_0)$ and 
$W_\pi^{\prime}M_\iota\tilde{\Sigma}_u(\theta_0) M_\iota W_\pi=\sum_{i=1}^n(W_{\pi(i)}-\bar{W})(W_{\pi(i)}-\bar{W})'\tilde{u}_i(\theta_0)^2=\sum_{i=1}^n(W_{i}-\bar{W})(W_{i}-\bar{W})'\tilde{u}_{\pi^{-1}(i)}(\theta_0)^2$, where $\bar{W}\equiv n^{-1}W'\iota$ and $\pi^{-1}$ is 
the inverse of $\pi$ i.e. $\pi\circ\pi^{-1}$ is the identity permutation, and it is uniformly distributed over $\mathbb{G}_n$. When $Z$ and $u(\theta_0)$ are independent under $H_0:\theta=\theta_0$, the PAR tests are still exact since Proposition \ref{ExactAR} holds for any $k$ provided the statistics are well-defined. 

Let $P_{ij}\equiv Z_i'(Z'Z)^{-1}Z_j$ denote the $(i,j)$ element of $P_Z$, $W_{ij}, j=1,\dots, k,$ be the $j$-th element of $W_i$, $\sum_{i\neq j}$ denote the summation over distinct values of the indices $i, j=1,\dots, n$, and $O_{a.s.}(\cdot)$ abbreviate $O(\cdot)$ almost surely. The asymptotic validity of the PAR tests are established under the following conditions in addition to Assumption \ref{A1}\ref{sa},\ref{or} and \ref{in}.

\begin{assumption}\label{A2}~ As $k,n\to\infty$, 
	\begin{enumerate}[label=(\alph*)]
		\item\label{A2a} $\max_{1\leq i\leq n}\left(\sum_{j=1, j\neq i}^n\vert P_{ij}\vert \right)=O_{a.s.}\left(\frac{\sum_{i\neq j}P_{ij}^2}{n\max_{i,j, i\neq j} |P_{ij}|}\right)$,
		\item\label{A2b} $\frac{\sum_{i\neq j}P_{ij}^2}{n^2\max_{i,j, i\neq j} P_{ij}^2}\conas 0$,
		\item\label{A2c} $k^{-1}\sum_{i=1}^nP_{ii}^2\conas 0$,
		\item\label{A2d} $\sup_{j\leq k}(\E_P[W_{1j}^4u_1^4]+\E_P[W_{1j}^4]+\E_P[u_1^4])<M_0<\infty$. 
	\end{enumerate}
\end{assumption}
Assumptions \ref{A2} is key for deriving the asymptotic distribution of the PAR statistics under the diverging-$k$ asymptotics. 
Assumptions \ref{A2}\ref{A2a} states that the row sums of the off-diagonal elements of $P_Z$ are roughly equal, Assumptions \ref{A2}\ref{A2b} requires that the off-diagonal elements of $P_Z$ are small relative to its largest element. 
By standard arguments, $\sum_{i=1}^nP_{ii}^2
\leq \sum_{i=1}^nP_{ii}=k$, hence $k^{-1}\sum_{i=1}^nP_{ii}^2\leq 1$. Assumption 
\ref{A2}\ref{A2c} makes a stronger requirement that the left-hand side of the latter inequality actually converges to $0$.\par 
To gain insight into Assumption \ref{A2}\ref{A2a}-\ref{A2c}, for simplicity, let $Z$ be a matrix of group dummy variables i.e. the $j$-th column of $Z$ has entries 
$Z_{ij}=1, i=\sum_{l=1}^jn_{l-1}+1,\dots, \sum_{l=1}^jn_{l-1}+n_j, j=1,\dots, k,$ with $n_0=0, n=\sum_{j=1}^kn_j$, and $0$ otherwise. Let $n_{\min}\equiv \min_{1\leq j\leq k}n_j$ and $n_{\max}\equiv \max_{1\leq j\leq k}n_j$. By direct calculations, one can see that Assumption \ref{A2}\ref{A2a} is equivalent to $(1-n_{\max}^{-1})n n_{\min}^{-1}/(k-\sum_{j=1}^kn_j^{-1})=O_{a.s.}(1)$, which holds if $n_{\max}/n_{\min}=O_{a.s.}(1)$ and $n_{\min}\geq 2$ almost surely. Assumption \ref{A2}\ref{A2b} is equivalent to $(k-\sum_{j=1}^k n_j^{-1})/(n^2/n_{\min}^2)\conas 0$. The latter holds because 
$(k-\sum_{j=1}^k n_j^{-1})/(n^2/n_{\min}^2)\leq k^{-2}(k-\sum_{j=1}^k n_j^{-1})\leq k^{-2}(k-n^{-1}k^2)\leq k^{-1}\conas 0$ as $k\to\infty$ using 
$(n/n_{\min})\geq k$ and $\sum_{j=1}^k n_j^{-1}\geq k^2/(\sum_{j=1}^kn_j),$ which follows from the convexity of $x\mapsto x^{-1}$. Finally, \ref{A2}\ref{A2c} holds if $n_{\min}\conas \infty$, since $k^{-1}\sum_{i=1}^nP_{ii}^2=k^{-1}\sum_{j=1}^kn_j^{-1}\leq n_{\min}^{-1}\conas 0$. 

\par 

Assumption \ref{A2}\ref{A2d} replaces Assumption \ref{A1}\ref{mo}. 
The boundedness of the cross moment $\E_P[W_{1j}^4u_1^4]$ is needed because we allow for a growing $k$ which, in turn, imposes a restriction on the consistency of the covariance matrix estimator. Furthermore, the marginal moment conditions in Assumption \ref{A2}\ref{A2d} are slightly weaker than Assumption \ref{A1}\ref{mo} 
because we only show the tests have asymptotically correct level, as opposed to size (uniform validity). The main results of this section is as follows. 
\begin{proposition}\label{prop: ManyIV}
Let Assumptions \ref{A1}\ref{sa},\ref{or},\ref{in}, and \ref{A2} hold with $P$ independent of $n$. Assume that $k, N\to\infty$ and $k^3/n\to 0$ as $n\to\infty$. If $H_0:\theta=\theta_0$ is true, then, for $\alpha\in (0,1)$ and the permutation statistics $\mathrm{PAR}\in\{\mathrm{PAR}_1, \mathrm{PAR}_2\}$, it holds that 
\begin{equation*}
	\lim_{n\to \infty}\E_{P}[\phi_n^{\mathrm{PAR}}]=\alpha.
\end{equation*}
\end{proposition}
The asymptotic distribution of the PAR statistics is obtained as follows. 
Letting 
$\tilde{\sigma}^2\equiv n^{-1}\tilde{u}(\theta_0)'\tilde{u}(\theta_0)$, we can rewrite 
	\begin{align}\label{eq: PAR decomp}
		\mathrm{PAR}_2
		&=\tilde{u}_\pi(\theta_0)'Z
		\left(\left(\sum_{i=1}^nZ_iZ_i'\tilde{u}_{\pi(i)}(\theta_0)^2\right)^{-1}-(Z'Z)^{-1}\tilde{\sigma}^{-2}\right)Z'\tilde{u}_\pi(\theta_0)\notag\\
&\quad+\sum_{i=1}^nP_{ii}\tilde{u}_{\pi(i)}(\theta_0)^2\tilde{\sigma}^{-2}+\sum_{i\neq j}P_{ij}\tilde{u}_{\pi(i)}(\theta_0)\tilde{u}_{\pi(j)}(\theta_0)\tilde{\sigma}^{-2}.
	\end{align}	
	It can be shown that the first summand in the above expression is negligible while the second summand converges to $k$. 
	The key insight for establishing the asymptotic distribution of the third summand is to view it as 
	a double-indexed permutation statistic of the form $\sum_{i\neq j}a_{ij}b_{\pi(i)\pi(j)}$ to which one can apply a 
	CLT established by \cite{Pham-Mocks-Sroka(1989)}.\par 
	 
	Recently, \cite{Mikusheva-Sun(2022)} propose a heteroskedasticity-robust Jackknife AR test 
	that is based on $\sum_{i\neq j}P_{ij}\tilde{u}_{i}(\theta_0)\tilde{u}_{j}(\theta_0)$ and is asymptotically valid 
	under the condition $k/n\to \tau\in(0,1)$ among others. Our proof strategy is likely to be useful for deriving the asymptotic distribution of a permutation version of their statistic. 
We leave this extension as well as those of the PLM and PCLR statistics under conditions less stringent than those 
in Assumption \ref{A2} to future research.
	\section{Simulations}\label{sec: Sim}
	This section presents a simulation evidence on the performance of the proposed 
	permutation tests. The data were generated according to 
	\begin{align*}
		y_i&= Y_i\theta+\gamma_1+X_{2i}^{\prime}\gamma_2+u_i,\\
		Y_i&=W_i^{\prime}\mathit{\Gamma}+\psi_1+X_{2i}^{\prime}\mathit{\Psi}_2+V_i, \quad i=1,\ldots,n,
	\end{align*}
	where $\theta=0$, $\mathit{\Gamma}=(1,\dots, 1)^{\prime}\sqrt{\lambda/(nk)}\,(k\times 1)$ and $\lambda\in\{0.1, 4, 20\}$, $X_{2i}$ is a $(p-1)$-vector of included non-constant exogenous variables with $p\in\{1, 5\}$, and $V_i=\rho u_i+\sqrt{1-\rho^2}\epsilon_i$ with $\rho=0.5$. The parameters 
	$\gamma_1, \gamma_2, \psi_1$, and $\mathit{\Psi}_2$ are set equal to $0$. As specified below, $W_i$ has zero mean and unit covariance matrix except for the first part of simulations in Section \ref{Htails}, hence  
	$\lambda=n\mathit{\Gamma}^{\prime}\E[W_iW_i^{\prime}]\mathit{\Gamma}\approx\mathit{\Gamma}^{\prime}W^{\prime}W\mathit{\Gamma}$. We consider the values $\lambda\in\{0.1, 4, 20\}$ which correspond to very weak, weak and strong identification. The tested restriction is $H_0:\theta=\theta_0=0$.\par 
	We implement the heteroskedasticity-robust asymptotic tests denoted as AR, LM, CLR, their permutation versions PAR$_1$, PAR$_2$, PLM, PCLR, the normal score and Wilcoxon score rank-based AR tests of \cite{Andrews-Marmer(2008)} denoted as RARn and RARw, respectively, the normal score and Wilcoxon score rank-based CLR statistics of \cite{Andrews-Soares(2007)} denoted as RCLRn and RCLRw, respectively, 
	and the wild bootstrap 
	AR and LM tests of \cite{Davidson-MacKinnon(2012)} denoted as WAR and WLM, respectively. 
	For both the asymptotic and permutation CLR-type tests, we do not make the  eigenvalue-adjustment (i.e. $\epsilon=0$ in footnote \ref{evaladj}).\par 
	In all simulations below, the null rejection probabilities of the tests (including CLR, PCLR, the rank-based and the wild bootstrap tests considered below) are computed using 2000 replications with 999 simulated samples for each replication while we display some power curves associated with the tests in Supplemental Appendix. We consider the following three cases in turn. 
	\subsection{Heavy tails}\label{Htails}
	To examine the finite sample validity, we consider heavy-tailed observations where each element of the random vector $(W_i^{\prime}, X_{2i}^{\prime}, u_i, \epsilon_i)^{\prime}$ is drawn independently from standard Cauchy distribution, and set  $k\in\{5, 10\}$, $n\in\{50, 100\}$, and $\lambda=4$.\par   
	Table \ref{HT} presents the null rejection probabilities. The asymptotic tests all underreject. 
	The PAR$_1$ test has nearly correct levels as predicted by Theorem \ref{ExactAR} in all cases, and the $\mathrm{PAR}_2$ test, which is exact only when 
	$p=1$, has slightly more accurate rejection rates in the case $p=1$ than in the case $p=5$. 
	The RARn, RARw, RCLRn and RCLRw tests are all robust against heavy-tailed errors when the IVs, the exogenous covariates and the error terms are independent as stated in Assumption \ref{ExactAR}. 
	The RARn and RARw tests are exact and this is borne out in the simulation results. 
	The RCLRn and RCLRw tests slightly underreject which may be attributed to their asymptotic nature.\par  
	Note, however, that, under the independence assumption and heavy-tailed errors, the rank-based tests have better power properties than the asymptotic tests, see \cite{Andrews-Soares(2007)} and \cite{Andrews-Marmer(2008)}.   
	In such cases, the permutation tests are likely to be dominated by the rank-based tests in terms of power given the local asymptotic equivalence between the permutation tests and the asymptotic tests under strong identification 
	established in Proposition \ref{LP}.   
	\subsection{Independence vs. standard exclusion restriction}\label{subsec: ISER}
	The next set of simulations highlights the difference between the independence and the standard exclusion restriction assumption in homoskedastic setting. 
	Two cases are considered: 
	\begin{enumerate}
		\item $(W_i^{\prime}, X_{2i}^{\prime}, u_i, \epsilon_i)^{\prime}\sim t_{5}[0, I_{k+p+1}]$, where $t_{5}[0, I_{k+p+1}]$ stands for $(k+p+1)$-variate $t$-distribution with degrees of freedom $5$ and covariance matrix $I_{k+p+1}$. Under this setting, these random variables have finite fourth moments, and $(W_i^{\prime}, X_{2i}^{\prime})^{\prime}$ satisfies the standard exclusion restriction in Assumption \ref{A1}: $\E[(u_i, V_i^{\prime})^{\prime}(W_i^{\prime}, X_{i}^{\prime})]=0$ but $(u_i, \epsilon_i)^{\prime}$ and $(W_i^{\prime}, X_{i}^{\prime})^{\prime}$ are dependent;\footnote{Because the distribution is fixed i.e. it does not vary with the sample size, the permutation tests remain asymptotically valid under the finite fourth moment assumption.}
		\item $(W_i^{\prime}, X_{2i}^{\prime}, u_i, \epsilon_i)^{\prime}\sim N[0, I_{k+p+1}]$. Clearly, in this case the instruments and the error terms are independent.  
	\end{enumerate}
	The results for the overidentified models with $n=100$ and $k=5$ are reported in 
	Table \ref{TI2}. The permutation tests perform better than the asymptotic tests regardless of whether the 
	instruments and the error terms are independent or not. 
	The rank-based tests, displayed in the bottom part of Table \ref{TI2}, 
	have nearly correct level in the independent case but overreject in the dependent case.\par
	The results for the just-identified design $k=1$ are similar to the over-identified case, thus are relegated to 
	Supplemental Appendix. 
	\subsection{Conditional heteroskedasticity}\label{subsec: CHR}
	The next part of the simulations considers designs with conditional heteroskedasticity. 
	The data are generated according to 
	(\ref{SE}) and (\ref{RE}) with 
	\begin{align*}
		u_i
		&=W_{i1}\upsilon_{i},\\
		V_i
		&=\rho u_i+\sqrt{1-\rho^2}\epsilon_{i},
	\end{align*}
	where $W_{i1}$ is the first element of $W_i$, and similarly to the previous case 
	\begin{enumerate}
		\item $(W_i^{\prime}, X_{2i}^{\prime}, \upsilon_i, \epsilon_i)^{\prime}\sim t_{5}[0, I_{k+p+1}]$;
		\item $(W_i^{\prime}, X_{2i}^{\prime}, \upsilon_i, \epsilon_i)^{\prime}\sim N[0, I_{k+p+1}]$.
	\end{enumerate}
	The sample size and the number of instruments are $n=100$ and $k\in\{2, 5, 10\}$. 
	The results for $\lambda=4$ are displayed in Tables \ref{hett2}. The asymptotic AR test tends to underreject and 
	performs poorly compared to  
	the LM statistic. The permutation tests perform better than their asymptotic counterparts in most cases, and 
	on par with the wild bootstrap in the independent case. However, when the instruments satisfy the standard exogeneity condition, and there are included exogenous variables, the permutation tests appear to have an edge over the wild bootstrap AR test which overrejects. 
	Qualitatively, the same observations made in the homoskedastic case for the rank-based tests are also observed 
	in the heteroskedastic case as the rank-based tests reject by a substantial margin in the dependent case. 
	See Supplemental Appendix for additional simulation evidence for the cases $\lambda\in\{0.1, 20\}$.  
	\begin{table}[htbp]
		\small
		\caption{{Null rejection probabilities at $5\%$ level. Cauchy distributed IVs, controls and error terms. }}%
		\begin{center}
			\label{HT}
			\begin{tabular}{lcccccccc}
				\hline \hline
				$\lambda=4$ & \multicolumn{4}{c}{$n=50$}
				& \multicolumn{4}{c}{$n=100$}
				\\
				\cmidrule(lr){1-1}\cmidrule(lr){2-5}\cmidrule(lr){6-9}
				&\multicolumn{2}{c}{$p=1$} &\multicolumn{2}{c}{$p=5$} &\multicolumn{2}{c}{$p=1$} &\multicolumn{2}{c}{$p=5$}
				\\ 
				\cmidrule(lr){2-3}\cmidrule(lr){4-5}\cmidrule(lr){6-7}\cmidrule(lr){8-9}
				Tests & $k=5$ & $k=10$&  $k=5$ & $k=10$  & $k=5$ & $k=10$ & $k=5$ & $k=10$\\ 
				\cmidrule{1-1}\cmidrule(lr){2-3}\cmidrule(lr){4-5}\cmidrule(lr){6-7}\cmidrule(lr){8-9}
			AR & 0.95 & 0.55 & 0.85 & 0.25 & 0.50 & 1.05 & 0.25 & 0.40 \\ 
			LM & 2.35 & 4.25 & 2.05 & 3.10 & 2.10 & 4.45 & 1.85 & 3.45 \\ 
			CLR & 2.10 & 3.60 & 1.60 & 2.60 & 1.75 & 3.85 & 1.25 & 3.25 \\ 			
				\cmidrule{1-1}\cmidrule(lr){2-3}\cmidrule(lr){4-5}\cmidrule(lr){6-7}\cmidrule(lr){8-9}
	PAR$_1$ & 4.40 & 5.30 & 4.70 & 5.60 & 4.05 & 4.95 & 4.45 & 4.70 \\ 
	PAR$_2$ & 4.40 & 5.45 & 3.40 & 3.50 & 3.95 & 4.75 & 3.25 & 3.05 \\ 
		\cmidrule{1-1}\cmidrule(lr){2-3}\cmidrule(lr){4-5}\cmidrule(lr){6-7}\cmidrule(lr){8-9}
		RARn & 4.25 & 4.55 & 4.10 & 4.15 & 5.10 & 4.05 & 4.90 & 3.75 \\ 
		RARw & 4.10 & 4.85 & 4.30 & 4.55 & 5.15 & 3.95 & 4.55 & 3.60 \\ 
		RCLRn & 3.15 & 3.75 & 4.20 & 5.10 & 4.55 & 3.50 & 3.85 & 4.00 \\ 
		RCLRw & 2.35 & 3.25 & 3.35 & 4.25 & 3.50 & 2.80 & 3.25 & 3.45 \\  \hline \hline
			\end{tabular}
		\end{center}
		\footnotesize{Note: AR, LM, and CLR are the heteroskedasticity-robust asymptotic tests, PAR$_1$ and PAR$_2$ denote the robust permutation AR tests, RARn and RARw denote the normal score and Wilcoxon score rank-based AR tests of \cite{Andrews-Marmer(2008)}, and RCLRn and RCLRw denote the normal score and Wilcoxon score rank-based CLR tests of \cite{Andrews-Soares(2007)} respectively. The number of replications is 2000, and the number of permutation samples for each replication is $999\, (N=1000)$.}
	\end{table}
	\begin{table}[htbp]
		\small
		\caption{{Null rejection probabilities at $5\%$ level. Over-identified homoskedastic model with $k=5$, $n=100$.}}%
		\begin{center}
			\label{TI2}
			\setlength\tabcolsep{3pt} 
			\begin{tabular}{lcccccccccccc}
				\hline \hline
				$n=100$ &\multicolumn{4}{c}{$\lambda=0.1$} & \multicolumn{4}{c}{$\lambda=4$}
				&\multicolumn{4}{c}{$\lambda=20$} \\ 
				\cmidrule(lr){1-1}\cmidrule(lr){2-5}\cmidrule(lr){6-9}\cmidrule(lr){10-13}
				$k=5$
				&\multicolumn{2}{c}{$t_{5}[0, I_{k+p+1}]$} & \multicolumn{2}{c}{$N[0, I_{k+p+1}]$}
				&\multicolumn{2}{c}{$t_{5}[0, I_{k+p+1}]$} & \multicolumn{2}{c}{$N[0, I_{k+p+1}]$} 
				&\multicolumn{2}{c}{$t_{5}[0, I_{k+p+1}]$} & \multicolumn{2}{c}{$N[0, I_{k+p+1}]$}\\ 
				\cmidrule(lr){1-1}\cmidrule(lr){2-3}\cmidrule(lr){4-5}\cmidrule(lr){6-7}
				\cmidrule(lr){8-9}\cmidrule(lr){10-11}\cmidrule(lr){12-13}
				Tests &  $p=1$   & $p=5$      &  $p=1$   & $p=5$   &$p=1$   & $p=5$    &  $p=1$   & $p=5$   &$p=1$   & $p=5$ &  $p=1$   & $p=5$     \\ 
				\cmidrule(lr){1-1}\cmidrule(lr){2-3}\cmidrule(lr){4-5}\cmidrule(lr){6-7}
				\cmidrule(lr){8-9}\cmidrule(lr){10-11}\cmidrule(lr){12-13}
			AR & 3.15 & 4.50 & 4.50 & 4.25 & 3.15 & 4.50 & 4.50 & 4.25 & 3.15 & 4.50 & 4.50 & 4.25 \\ 
			LM & 2.90 & 4.30 & 4.15 & 4.25 & 3.20 & 4.05 & 3.70 & 3.70 & 3.00 & 4.10 & 3.25 & 3.85 \\ 
			CLR & 2.80 & 3.95 & 4.05 & 4.05 & 3.20 & 4.10 & 3.35 & 3.45 & 2.55 & 4.00 & 3.30 & 3.95 \\ 
			 \cmidrule(lr){1-1}\cmidrule(lr){2-3}\cmidrule(lr){4-5}\cmidrule(lr){6-7}
				\cmidrule(lr){8-9}\cmidrule(lr){10-11}\cmidrule(lr){12-13}
				PAR$_1$ & 5.00 & 5.55 & 5.45 & 4.75 & 5.00 & 5.55 & 5.45 & 4.75 & 5.00 & 5.55 & 5.45 & 4.75 \\ 
				PAR$_2$ & 5.05 & 6.85 & 5.60 & 5.45 & 5.05 & 6.85 & 5.60 & 5.45 & 5.05 & 6.85 & 5.60 & 5.45 \\ 
				PLM & 5.10 & 5.60 & 5.40 & 5.15 & 5.00 & 6.05 & 5.05 & 5.05 & 4.70 & 6.15 & 4.95 & 5.45 \\ 
				PCLR & 4.70 & 6.10 & 5.15 & 4.90 & 4.40 & 5.40 & 4.00 & 4.30 & 3.65 & 5.10 & 3.70 & 4.20 \\ 
				 \cmidrule(lr){1-1}\cmidrule(lr){2-3}\cmidrule(lr){4-5}\cmidrule(lr){6-7}
				\cmidrule(lr){8-9}\cmidrule(lr){10-11}\cmidrule(lr){12-13}
			RARn & 18.65 & 15.70 & 5.25 & 4.45 & 18.65 & 15.70 & 5.25 & 4.45 & 18.65 & 15.70 & 5.25 & 4.45 \\ 
			RARw & 12.40 & 10.20 & 5.35 & 4.10 & 12.40 & 10.20 & 5.35 & 4.10 & 12.40 & 10.20 & 5.35 & 4.10 \\ 
			RCLRn & 13.85 & 14.55 & 3.35 & 3.70 & 13.15 & 14.50 & 3.15 & 3.75 & 11.10 & 12.65 & 4.00 & 3.50 \\ 
			RCLRw & 10.40 & 12.90 & 4.45 & 4.55 & 10.20 & 11.90 & 4.20 & 4.65 & 9.15 & 10.90 & 4.55 & 4.35 \\ 
			 \cmidrule(lr){1-1}\cmidrule(lr){2-3}\cmidrule(lr){4-5}\cmidrule(lr){6-7}
				\cmidrule(lr){8-9}\cmidrule(lr){10-11}\cmidrule(lr){12-13} 
				WAR & 5.30 & 7.45 & 5.35 & 4.10 & 5.30 & 7.45 & 5.35 & 4.10 & 5.30 & 7.45 & 5.35 & 4.10 \\ 
				WLM & 4.65 & 5.85 & 4.95 & 4.80 & 4.20 & 6.40 & 4.40 & 4.05 & 4.50 & 5.95 & 4.30 & 4.35 \\ 
				 \hline \hline
			\end{tabular}
		\end{center}
		\footnotesize{Note: $t_{5}[0, I_{k+p+1}]$ and $N[0, I_{k+p+1}]$ correspond to $(W_i^{\prime}, X_{2i}^{\prime}, u_i, \epsilon_i)^{\prime}\sim t_{5}[0, I_{k+p+1}]$ (dependent but uncorrelated) and $(W_i^{\prime}, X_{2i}^{\prime}, u_i, \epsilon_i)^{\prime}\sim N[0, I_{k+p+1}]$ (independent), respectively. 
			PLM and PCLR denote the robust permutation LM and CLR tests, WAR and WLM are the wild bootstrap AR and LM tests of \cite{Davidson-MacKinnon(2012)} and the remaining tests are as in Table \ref{HT}. The number of replications is 2000, and the number of permutation samples for each replication is $999\, (N=1000)$.}		
	\end{table}
	\begin{table}[htbp]
		\small
		\caption{{
				Null rejection probabilities at $5\%$ level. Conditionally heteroskedastic design with weak identification ($\lambda=4$).
		}}%
		\begin{center}
			\label{hett2}
			\setlength\tabcolsep{3pt} 
			\begin{tabular}{lcccccccccccc}
				\hline\hline
				$n=100$ & \multicolumn{6}{c}{$p=1$} & \multicolumn{6}{c}{$p=5$}\\
				\cmidrule(lr){1-1}\cmidrule(lr){2-7}\cmidrule(lr){8-13}
				&\multicolumn{3}{c}{{\footnotesize{$t_{5}[0, I_{k+p+1}]$}}} & \multicolumn{3}{c}{{\footnotesize{$N[0, I_{k+p+1}]$}}} &\multicolumn{3}{c}{{\footnotesize{$t_{5}[0, I_{k+p+1}]$}}} & \multicolumn{3}{c}{{\footnotesize{$N[0, I_{k+p+1}]$}}} \\ 
				\cmidrule(lr){2-4}\cmidrule(lr){5-7}\cmidrule(lr){8-10}\cmidrule(lr){11-13}
				Tests &    $k=2$   & $k=5$      & $k=10$      &$k=2$      & $k=5$      &$k= 10$ &$k=2$   & $k=5$      & $k=10$      &$k=2$      & $k=5$      &$k= 10$      \\ 
				\cmidrule(lr){1-1}\cmidrule(lr){2-4}\cmidrule(lr){5-7}\cmidrule(lr){8-10}\cmidrule(lr){11-13}
				AR & 2.35 & 1.75 & 0.60 & 4.05 & 2.90 & 1.40 & 4.30 & 2.40 & 1.15 & 5.25 & 3.25 & 2.60 \\ 
				LM & 3.10 & 3.95 & 3.70 & 4.55 & 3.85 & 3.80 & 4.50 & 2.80 & 2.50 & 6.00 & 3.60 & 4.00 \\ 
				CLR & 2.40 & 2.25 & 1.50 & 4.10 & 3.15 & 1.85 & 4.15 & 2.45 & 1.30 & 5.65 & 3.00 & 3.45 \\ 
				
				\cmidrule(lr){1-1}\cmidrule(lr){2-4}\cmidrule(lr){5-7}\cmidrule(lr){8-10}\cmidrule(lr){11-13}
				PAR$_1$ & 3.60 & 5.15 & 4.35 & 4.85 & 5.60 & 5.20 & 7.40 & 7.55 & 8.85 & 5.95 & 4.70 & 4.70 \\ 
				PAR$_2$ & 3.35 & 4.85 & 4.35 & 4.90 & 5.30 & 5.00 & 6.35 & 5.50 & 4.95 & 6.40 & 5.30 & 5.35 \\ 
				PLM & 4.20 & 5.35 & 5.80 & 5.15 & 4.90 & 5.05 & 6.20 & 3.90 & 4.60 & 6.30 & 4.45 & 6.10 \\ 
				PCLR & 3.80 & 5.30 & 4.20 & 4.95 & 4.60 & 3.95 & 5.85 & 4.50 & 3.50 & 6.65 & 4.45 & 5.55 \\ 
				
				\cmidrule(lr){1-1}\cmidrule(lr){2-4}\cmidrule(lr){5-7}\cmidrule(lr){8-10}\cmidrule(lr){11-13}
				RARn & 22.30 & 27.30 & 32.55 & 13.45 & 10.20 & 8.05 & 23.55 & 24.80 & 29.85 & 13.75 & 10.10 & 8.70 \\ 
				RARw & 14.35 & 17.90 & 20.10 & 10.40 & 8.85 & 6.65 & 15.95 & 16.00 & 17.55 & 10.40 & 8.25 & 7.35 \\ 
				RCLRn & 19.45 & 22.40 & 29.95 & 11.05 & 7.50 & 6.05 & 29.25 & 31.25 & 41.50 & 11.85 & 8.50 & 7.15 \\ 
				RCLRw & 13.70 & 17.50 & 19.55 & 9.40 & 7.45 & 6.35 & 25.00 & 26.20 & 36.00 & 10.10 & 8.25 & 7.55 \\ 
				
				\cmidrule(lr){1-1}\cmidrule(lr){2-4}\cmidrule(lr){5-7}\cmidrule(lr){8-10}\cmidrule(lr){11-13}
			WAR & 4.50 & 5.80 & 5.55 & 4.75 & 5.20 & 4.90 & 8.65 & 10.80 & 12.60 & 5.75 & 4.40 & 4.95 \\ 
			WLM & 3.20 & 3.50 & 4.15 & 4.00 & 4.05 & 4.90 & 6.60 & 5.85 & 7.00 & 5.55 & 3.90 & 5.30 \\ 
				\hline \hline 
			\end{tabular}
		\end{center}
		\footnotesize{Note: $t_{5}[0, I_{k+p+1}]$ and $N[0, I_{k+p+1}]$ correspond to $(W_i^{\prime}, X_{2i}^{\prime}, \upsilon_i, \epsilon_i)^{\prime}\sim t_{5}[0, I_{k+p+1}]$ (dependent but uncorrelated) and $(W_i^{\prime}, X_{2i}^{\prime}, \upsilon_i, \epsilon_i)^{\prime}\sim N[0, I_{k+p+1}]$ (independent), respectively. The error terms are generated as $u_i=W_{i1}\upsilon_{i}, V_i=\rho u_i+\sqrt{1-\rho^2}\epsilon_{i}$, where $W_{i1}$ is the first element of $W_i$ and $\rho=0.5$. The tests are in Table \ref{TI2}. The number of replications is 2000, and the number of permutation samples for each replication is $N=999$.}
	\end{table}

	\section{Empirical application}\label{sec: EA}
	\cite{Bazzi-Clemens(2013)} re-examine cross-sectional IV regression results about the effect of foreign aid on economic growth from several published studies. The variables in the ``aid-growth" IV regressions estimated by \cite{Rajan-Subramanian(2008)} and \cite{Bazzi-Clemens(2013)} are 
	\begin{itemize}
		\item $y_i:$ the average annual growth of per capita GDP; 
		\item $Y_i:$ the foreign-aid receipts to GDP ratio (Aid/GDP); 
		\item $W_{i1}:$ a variable constructed from aid-recipient population size, aid-donor population size, colonial relationship, and language traits (see Appendix A of \cite{Bazzi-Clemens(2013)});
		\item $X_i:$ constant, initial per capita GDP, initial level of policy, initial level of life expectancy, geography, institutional quality, initial inflation, initial M2/GDP, initial budget balance/GDP, Revolutions, Ethnic fractionalization, Sub-Saharan Africa dummy and East Asia dummy.
	\end{itemize}
	The data cover the period 1970-2000 and the sample size is $n=78$.  
	\citet[Table 1, Columns 1-4]{Bazzi-Clemens(2013)} consider the following four specifications: 
	\begin{itemize}
		\item Specification 1: The baseline specification of \citet[Table 4, Column 2]{Rajan-Subramanian(2008)} with the variables as above; 
		\item Specification 2: log population is included in the second stage of Specification 1;
		\item Specification 3: $W_{i2}=$ log population replaces the instrument $W_{i1}$ in Specification 1;
		\item Specification 4: An instrument, $W_{i3}$, constructed from the colonial ties indicators only, replaces $W_{i1}$ in Specification 1. 
	\end{itemize}
	Specifications 1-4 are just-identified as there is a single instrument for the single endogenous regressor, hence 
	we only consider the AR and PAR test results. In addition, we consider Specifications 5-7 with IVs $(W_{i1}, W_{i2})'$, $(W_{i2}, W_{i3})'$ and $(W_{i1}, W_{i2}, W_{i3})'$, respectively, to examine whether the over-identified specifications could have any instrumentation power. $W_{i2}$ is included in the latter three cases because it is the strongest instrument as reported in \cite{Bazzi-Clemens(2013)}. We also include the confidence interval based on the two-stage least squares $t$-test, denoted as $t_{\mathrm{2sls}}$.\par 
	Tables \ref{BC13} and \ref{BC13b} report the $95\%$ confidence intervals for the coefficient on Aid/GDP. In all specifications, the Breusch-Pagan test $p$-values from the two separate reduced-form regressions of the endogenous variables on the exogenous variables show an evidence of heteroskedasticity. 
	
	The results for Specifications 1-4 in Table \ref{BC13} qualitatively agree with the homoskedastic CLR confidence intervals reported in \cite{Bazzi-Clemens(2013)}. The PAR$_2$ confidence intervals are nearly identical but slightly shorter than the asymptotic and wild bootstrap confidence intervals in Specifications 1 and 3. The very wide  confidence intervals in Specifications 2 and 4 that do not use log population or the instruments based on it indicate that the population size instrument used in the other specifications has indeed the most identifying power.\par 
	The above results also extend to Table \ref{BC13b}. Some noteworthy findings are as follows. The WAR confidence intervals are wider than the AR and PAR confidence intervals. The PLM confidence intervals are shorter than 
	the LM and WLM confidence intervals, but still wider than the other confidence intervals. Somewhat surprisingly, in Specification 6, the AR, PAR and $\mathrm{PCLR}$ confidence intervals exclude $0$ thereby pointing towards borderline significant effect of the foreign aid on growth despite the rather small value of the first-stage $F$-statistic 17.90. Also, in Specification 7, the $\mathrm{PAR}$ confidence intervals produce significant results which underscore the utility of the proposed tests.\par 
	However, these results should be interpreted with caution because the population size might affect the growth through multiple channels as argued by \cite{Bazzi-Clemens(2013)}.   
	Overall, the uncertainty regarding the effect of foreign aid in the sample data 
	may be expressed more accurately by the permutation and the wild bootstrap confidence intervals 
	as the unobserved confounders in the ``aid-growth" regression are more likely to be orthogonal to, rather than independent of the population size. 
	\begin{table}[htbp]
		\small
		\caption{{Confidence intervals for the coefficient on Aid/GDP in \cite{Rajan-Subramanian(2008)} and \cite{Bazzi-Clemens(2013)} data.}}%
		\begin{center}
			\label{BC13}
			\setlength\tabcolsep{2pt} 
			\begin{tabular}{lcccccc}
				\hline \hline
				&\multicolumn{2}{c}{Spec. 1}& \multicolumn{1}{c}{Spec. 2}& \multicolumn{2}{c}{Spec. 3}&\multicolumn{1}{c}{Spec. 4}\\
				Test statistics & $95\%$ CI & Length &$95\%$ CI &$95\%$ CI& Length &$95\%$ CI\\ 
				\cmidrule(lr){1-1}\cmidrule(lr){2-3}\cmidrule(lr){4-4}\cmidrule(lr){5-6}\cmidrule(lr){7-7}
			
				$t_{\mathrm{2sls}}$ & $[-0.05, 0.24]$ & $0.29$ & $[-5.26, 7.08]$  & $[-0.06, 0.21]$ & $0.27$ & $[-991.83, 959.94]$\\ 
				 Hom. AR & $[-0.02, 0.28]$ & $0.30$ & $\supseteq[-40, 40]$ & $[-0.03, 0.25]$ & $0.28$ & $\supseteq[-40, 40]$\\ 
				 
				\cmidrule(lr){1-1}\cmidrule(lr){2-3}\cmidrule(lr){4-4}\cmidrule(lr){5-6}\cmidrule(lr){7-7}
				
				RARn & $[-0.05, 0.30]$ & $0.35$ & $\supseteq[-40, 40]$  & $[-0.06, 0.27]$ & $0.33$ & $\supseteq[-40, 40]$ \\ 
				RARw & $[-0.06, 0.17]$ & $0.23$ & $\supseteq[-40, 40]$ & $[-0.07, 0.17]$ & $0.24$ & $\supseteq[-40, 40]$ \\ 
				RCLRn & $[-0.03, 0.25]$ & $0.28$ & $\supseteq[-40, 40]$ & $[-0.04, 0.23]$ & $0.27$ & $\supseteq[-40, 40]$ \\ 
				RCLRw & $[-0.04, 0.18]$ & $0.22$ & $\supseteq[-40, 40]$ & $[-0.05, 0.15]$ & $0.20$ & $\supseteq[-40, 40]$\\ 
				
				\cmidrule(lr){1-1}\cmidrule(lr){2-3}\cmidrule(lr){4-4}\cmidrule(lr){5-6}\cmidrule(lr){7-7}
		
				  AR & $[-0.01, 0.29]$ & 0.30 & $\supseteq[-40, 40]$ & $[-0.02, 0.25]$ & $0.27$ & $\supseteq[-40, 40]$\\ 
				
				PAR$_1$ & $[-0.02, 0.30]$ & $0.32$ & $\supseteq[-40, 40]$ & $[-0.03, 0.27]$ & $0.30$ & $\supseteq[-40, 40]$\\ 
				PAR$_2$ & $[-0.01, 0.28]$ & $0.29$ & $\supseteq[-40, 40]$ & $[-0.02, 0.24]$ & $0.26$ & $\supseteq[-40, 40]$\\

				\cmidrule(lr){1-1}\cmidrule(lr){2-3}\cmidrule(lr){4-4}\cmidrule(lr){5-6}\cmidrule(lr){7-7}
			
				WAR & $[-0.02, 0.31]$ & $0.33$ & $\supseteq[-40, 40]$ & $[-0.03, 0.27]$ & $0.30$ & $\supseteq[-40, 40]$\\ 
				\cmidrule(lr){1-1}\cmidrule(lr){2-3}\cmidrule(lr){4-4}\cmidrule(lr){5-6}\cmidrule(lr){7-7}
				$F$-statistic &\multicolumn{2}{c}{31.63}& $0.13$ &\multicolumn{2}{c}{36.30}&$0.00$\\
				Breusch-Pagan $p$-values &\multicolumn{2}{c}{\{0.03, 0.00\}}&\{0.03, 0.00\} &\multicolumn{2}{c}{\{0.02, 0.00\}}&\{0.01, 0.00\}\\ 
				\hline \hline 
			\end{tabular}
		\end{center}
		\footnotesize{Note: The number of permutation samples is $1999\, (N=2000)$. The sample size is $78$. Specifications 1-4 are just-identified. $t_{\mathrm{2sls}}$ and Hom. AR are the homoskedastic two-stage least squares $t$ and AR tests, and the remaining tests are as in Tables \ref{HT} and \ref{TI2}. The Breusch-Pagan test $p$-values are computed using the fitted values in the reduced-form regressions of $y_i$ and $Y_i$ on the instrument and exogenous covariates respectively.}
	\end{table}
	\begin{table}[htbp]
		\footnotesize
		\caption{{Confidence intervals for the coefficient on Aid/GDP in \cite{Rajan-Subramanian(2008)} and \cite{Bazzi-Clemens(2013)} data.}}%
		\begin{center}
			\label{BC13b}
			\setlength\tabcolsep{1.5pt} 
			\begin{tabular}{lcccccc}
				\hline\hline 
				&\multicolumn{2}{c}{Spec. 5}& \multicolumn{2}{c}{Spec. 6}& \multicolumn{2}{c}{Spec. 7}\\
				Test statistics & $95\%$ CI & Length &$95\%$ CI& Length &$95\%$ CI& Length\\ 
				\cmidrule(lr){1-1}\cmidrule(lr){2-3}\cmidrule(lr){4-5}\cmidrule(lr){6-7}
			
				$t_{\mathrm{2sls}}$ & $[-0.05, 0.22]$ & $0.28$ & $[-0.05, 0.24]$ & $0.29$ & $[-0.05, 0.24]$ & $0.29$ \\ 
				Hom. AR & $[-0.05, 0.31]$ & $0.36$ & $[-0.02, 0.26]$ & $0.28$ & $[-0.04, 0.32]$ & $0.36$ \\ 
				Hom. LM & $[-0.63, 0.26]$ & $0.89$ & $[-0.63, 0.28]$ & $0.91$ & $[-0.63, 0.28]$ & $0.91$ \\ 
				Hom. CLR & $[ -0.03 , 0.26 ]$ & $0.30$ & $[ -0.03 , 0.28 ]$ & $0.31$ & $[ -0.03 , 0.28 ]$ & $0.31$ \\ 
				
				\cmidrule(lr){1-1}\cmidrule(lr){2-3}\cmidrule(lr){4-5}\cmidrule(lr){6-7}
	
				RARn & $[-0.08, 0.41]$ & $0.49$ & $[-0.06, 0.32]$ & $0.38$ & $[-0.09, 0.4]$ & $0.49$ \\ 
				RARw & $[-0.08, 0.29]$ & $0.37$ & $[-0.06, 0.23]$ & $0.29$ & $[-0.09, 0.27]$ & $0.36$ \\ 
				RCLRn & $[-0.03, 0.24]$ & $0.27$ & $[-0.03, 0.25]$ & $0.28$ & $[-0.03, 0.25]$ & $0.28$ \\ 
				RCLRw & $[-0.05, 0.17]$ & $0.22$ & $[-0.03, 0.17]$ & $0.20$ & $[-0.04, 0.17]$ & $0.21$ \\

				 \cmidrule(lr){1-1}\cmidrule(lr){2-3}\cmidrule(lr){4-5}\cmidrule(lr){6-7}
				
				AR & $[-0.04, 0.34]$ & $0.38$ & $[0.04, 0.32]$ & $0.28$ & $[0.01, 0.42]$ & $0.41$ \\ 
				LM & $[-0.43, 0.27]$ & $0.70$ & $[-0.61, 0.36]$ & $0.97$ & $[-0.69, 0.35]$ & $1.04$ \\ 
				CLR & $[-0.02, 0.28]$ & $0.30$ & $[0.00, 0.36]$ & $0.36$ & $[-0.03, 0.34]$ & $0.37$ \\ 
				
				\cmidrule(lr){1-1}\cmidrule(lr){2-3}\cmidrule(lr){4-5}\cmidrule(lr){6-7}
			
				PAR$_1$ & $[-0.05, 0.34]$ & $0.39$ & $[0.04, 0.34]$ & $0.30$ & $[0.01, 0.43]$ & $0.42$ \\ 
				PAR$_2$ & $[-0.04, 0.34]$ & $0.38$ & $[0.04, 0.31]$ & $0.27$ & $[0.02, 0.38]$ & $0.36$ \\ 
				PLM & $[-0.43, 0.25]$ & $0.68$ & $[-0.6, 0.33]$ & $0.93$ & $[-0.65, 0.33]$ & $0.98$ \\ 
				PCLR & $[-0.03, 0.27]$ & $0.30$ & $[0.01, 0.34]$ & $0.33$ & $[-0.02, 0.35]$ & $0.37$ \\ 
				
				\cmidrule(lr){1-1}\cmidrule(lr){2-3}\cmidrule(lr){4-5}\cmidrule(lr){6-7}
		
				WAR & $[-0.05, 0.44]$ & $0.49$ & $[0, 0.33]$ & $0.33$ & $[-0.04, 0.48]$ & $0.52$ \\ 
				WLM & $[-0.64, 0.30]$ & $0.94$ & $[-0.66, 0.30]$ & $0.96$ & $[-0.66, 0.33]$ & $0.99$ \\ 
				
				\cmidrule(lr){1-1}\cmidrule(lr){2-3}\cmidrule(lr){4-5}\cmidrule(lr){6-7}
				$F$-statistic &\multicolumn{2}{c}{17.97}&\multicolumn{2}{c}{17.90}&\multicolumn{2}{c}{11.86}\\
				Sargan $p$-value &\multicolumn{2}{c}{0.42}&\multicolumn{2}{c}{0.08}& \multicolumn{2}{c}{0.21}\\	
				Breusch-Pagan $p$-values &\multicolumn{2}{c}{\{0.03, 0.00\}}&\multicolumn{2}{c}{\{0.02, 0.00\}}&  \multicolumn{2}{c}{\{0.02, 0.00\}}\\	
				\hline \hline 
			\end{tabular}
		\end{center}
		\footnotesize{Note: The number of permutation samples is $1999\, (N=2000)$. The sample size is $78$. Specifications 5-7 are over-identified. Hom. LM and Hom. CLR are the homoskedastic LM and CLR tests, Sargan $p$-value is the $p$-value associated with the Sargan test of overidentifying restrictions, and the remaining tests are as in Tables \ref{BC13}.}
	\end{table}
	\section{Conclusion}\label{Conc}
	The PAR$_1$ test shares the strengths of the rank/permutation tests and the wild bootstrap tests because it is exact under independence and heteroskedasticity-robust under standard assumptions.  
	This test is recommended when the IVs are randomly assigned, hence independent of the exogenous covariates and the error term. The PAR$_2$ test also shows a satisfactory performance in the simulations. Because the 
	PAR tests do not require the correct specification of the first-stage equation as they directly uses the IVs in the equation $\hat{m}(\theta_0)=n^{-1}Z'\tilde{u}(\theta_0)$ \citep{Dufour-Taamouti(2007)}, they have robustness advantages relative to the other permutation tests.
\par  
Moreover, the PCLR test is found to perform reasonably well in the simulations. Assuming that \eqref{SE}-\eqref{RE} are correctly specified, as we have shown, the PCLR test is similar and equivalent to 
	the CLR tests of \cite{Andrews-Guggenberger(2019)} asymptotically. In homoskedastic IV models, the latter tests are asymptotically equivalent or reduce to the CLR test of \cite{Moreira(2003)} which is known to be optimal when there is single endogenous regressor \citep{Andrews-Moreira-Stock(2006)}. However, despite being equivalent to the PCLR test under strong identification, the PLM test shows an inferior performance in terms of power. Based on the properties mentioned and the simulation evidence, we recommend reporting both the PAR and PCLR test results in practice.

\section*{Supplementary materials}
\begin{description}
\item[Supplemental Appendix:] This file contains all the proofs of the results presented in the paper, along with additional numerical evidence. (pdf)
\item[Replication Package:] This file contains all the materials needed to replicate the numerical results presented in the paper, including R script files for simulations and empirical applications, accompanying datasets, and replication instructions. (zip)
\end{description}

\section*{Acknowledgements}
The author would like to thank the editor and the three anonymous referees for their constructive comments, which have significantly improved the quality of this paper. The author also thanks Dmitry Arkhangelsky, Xavier D'Haultf\oe{}uille, Antoine Djogbenou, Hiroyuki Kasahara, Uros Petronijevic, and the participants of the Econometric Society European Winter Meeting 2019 and the 37th Canadian Econometrics Study Group (CESG) Meeting 2021 for helpful comments.

\section*{Funding}	
The author gratefully acknowledges financial support from the LA\&PS Minor Research Grant, York University.

\section*{Disclosure statement}
The author reports that there are no competing interests to declare.

	\small
	\bibliographystyle{ims}
	\bibliography{Permutation}

\begin{thebibliography}{47}
\expandafter\ifx\csname natexlab\endcsname\relax\def\natexlab#1{#1}\fi
\expandafter\ifx\csname url\endcsname\relax
  \def\url#1{\texttt{#1}}\fi
\expandafter\ifx\csname urlprefix\endcsname\relax\def\urlprefix{URL }\fi
\providecommand{\eprint}[2][]{\url{#2}}

\bibitem[{Anderson and Rubin(1949)}]{Anderson-Rubin(1949)}
\textsc{Anderson, T.~W.} and \textsc{Rubin, H.} (1949).
\newblock {Estimation of the Parameters of a Single Equation in a Complete
  System of Stochastic Equations}.
\newblock \textit{Annals of Mathematical Statistics}, \textbf{20} 46--63.

\bibitem[{Andrews and
  Guggenberger(2019{\natexlab{a}})}]{Andrews-Guggenberger(2019)}
\textsc{Andrews, D.~W.} and \textsc{Guggenberger, P.} (2019{\natexlab{a}}).
\newblock {Identification-and Singularity-Robust Inference for Moment Condition
  Models}.
\newblock \textit{Quantitative Economics}, \textbf{10} 1703--1746.

\bibitem[{Andrews et~al.(2020)Andrews, Cheng and
  Guggenberger}]{Andrews-Cheng-Guggenberger(2020)}
\textsc{Andrews, D. W.~K.}, \textsc{Cheng, X.} and \textsc{Guggenberger, P.}
  (2020).
\newblock {Generic Results for Establishing the Asymptotic Size of Confidence
  Sets and Tests}.
\newblock \textit{Journal of Econometrics}, \textbf{218} 496--531.

\bibitem[{Andrews and Guggenberger(2017)}]{Andrews-Guggenberger(2017)}
\textsc{Andrews, D. W.~K.} and \textsc{Guggenberger, P.} (2017).
\newblock {Asymptotic Size of Kleibergen's LM and Conditional LR Tests for
  Moment Condition Models}.
\newblock \textit{Econometric Theory}, \textbf{33} 1046--1080.

\bibitem[{Andrews and
  Guggenberger(2019{\natexlab{b}})}]{Andrews-Guggenberger(2019s)}
\textsc{Andrews, D. W.~K.} and \textsc{Guggenberger, P.} (2019{\natexlab{b}}).
\newblock {Supplemental material to ``Identification-and Singularity-Robust
  Inference for Moment Condition Models"}.
\newblock \textit{Quantitative Economics}, \textbf{10}.

\bibitem[{Andrews and Marmer(2008)}]{Andrews-Marmer(2008)}
\textsc{Andrews, D. W.~K.} and \textsc{Marmer, V.} (2008).
\newblock {Exactly Distribution-Free Inference in Instrumental Variables
  Regression with Possibly Weak Instruments}.
\newblock \textit{Journal of Econometrics}, \textbf{142} 183--200.

\bibitem[{Andrews et~al.(2006)Andrews, Moreira and
  Stock}]{Andrews-Moreira-Stock(2006)}
\textsc{Andrews, D. W.~K.}, \textsc{Moreira, M.~J.} and \textsc{Stock, J.~H.}
  (2006).
\newblock {Optimal Two-Sided Invariant Similar Tests for Instrumental Variables
  Regression}.
\newblock \textit{Econometrica}, \textbf{74} 715--752.

\bibitem[{Andrews and Soares(2007)}]{Andrews-Soares(2007)}
\textsc{Andrews, D. W.~K.} and \textsc{Soares, G.} (2007).
\newblock {Rank Tests for Instrumental Variables Regression with Weak
  Instruments}.
\newblock \textit{Econometric Theory}, \textbf{23} 1033--1082.

\bibitem[{Andrews and Stock(2007)}]{Andrews-Stock(2007)}
\textsc{Andrews, D. W.~K.} and \textsc{Stock, J.~H.} (2007).
\newblock {Inference with Weak Instruments}.
\newblock In \textit{Advances in Econometrics: Proceedings of the Ninth World
  Congress of the Econometric Society}, vol.~3.

\bibitem[{Andrews et~al.(2019)Andrews, Stock and Sun}]{Andrews-Stock-Sun(2019)}
\textsc{Andrews, I.}, \textsc{Stock, J.~H.} and \textsc{Sun, L.} (2019).
\newblock {Weak Instruments in Instrumental Variables Regression: Theory and
  Practice}.
\newblock \textit{Annual Review of Economics}, \textbf{11} 727--753.

\bibitem[{Bazzi and Clemens(2013)}]{Bazzi-Clemens(2013)}
\textsc{Bazzi, S.} and \textsc{Clemens, M.~A.} (2013).
\newblock {Blunt Instruments: Avoiding Common Pitfalls in Identifying the
  Causes of Economic Growth}.
\newblock \textit{American Economic Journal: Macroeconomics}, \textbf{5}
  152--86.

\bibitem[{Bugni et~al.(2018)Bugni, Canay and Shaikh}]{Bugni-Canay-Shaikh(2018)}
\textsc{Bugni, F.~A.}, \textsc{Canay, I.~A.} and \textsc{Shaikh, A.~M.} (2018).
\newblock Inference under covariate-adaptive randomization.
\newblock \textit{Journal of the American Statistical Association},
  \textbf{113} 1784--1796.

\bibitem[{Canay and Kamat(2018)}]{Canay-Kamat(2018)}
\textsc{Canay, I.~A.} and \textsc{Kamat, V.} (2018).
\newblock {Approximate Permutation Tests and Induced Order Statistics in the
  Regression Discontinuity Design}.
\newblock \textit{The Review of Economic Studies}, \textbf{85} 1577--1608.

\bibitem[{Canay et~al.(2017)Canay, Romano and
  Shaikh}]{Canay-Romano-Shaikh(2017)}
\textsc{Canay, I.~A.}, \textsc{Romano, J.~P.} and \textsc{Shaikh, A.~M.}
  (2017).
\newblock {Randomization Tests under an Approximate Symmetry Assumption}.
\newblock \textit{Econometrica}, \textbf{85} 1013--1030.

\bibitem[{Chernozhukov et~al.(2009)Chernozhukov, Hansen and
  Jansson}]{Chernozhukov-Hansen-Jansson(2009)}
\textsc{Chernozhukov, V.}, \textsc{Hansen, C.} and \textsc{Jansson, M.} (2009).
\newblock {Finite Sample Inference for Quantile Regression Models}.
\newblock \textit{Journal of Econometrics}, \textbf{152} 93--103.

\bibitem[{Chung and Romano(2013)}]{Chung-Romano(2013)}
\textsc{Chung, E.} and \textsc{Romano, J.~P.} (2013).
\newblock {Exact and Asymptotically Robust Permutation Tests}.
\newblock \textit{Annals of Statistics}, \textbf{41} 484--507.

\bibitem[{Chung and Romano(2016)}]{Chung-Romano(2016)}
\textsc{Chung, E.} and \textsc{Romano, J.~P.} (2016).
\newblock {Multivariate and Multiple Permutation Tests}.
\newblock \textit{Journal of Econometrics}, \textbf{193} 76--91.

\bibitem[{Davidson and MacKinnon(2008)}]{Davidson-MacKinnon(2008)}
\textsc{Davidson, R.} and \textsc{MacKinnon, J.~G.} (2008).
\newblock {Bootstrap Inference in a Linear Equation Estimated by Instrumental
  Variables}.
\newblock \textit{The Econometrics Journal}, \textbf{11} 443--477.

\bibitem[{Davidson and MacKinnon(2012)}]{Davidson-MacKinnon(2012)}
\textsc{Davidson, R.} and \textsc{MacKinnon, J.~G.} (2012).
\newblock {Wild Bootstrap Tests for IV Regression}.
\newblock \textit{Journal of Business \& Economic Statistics}, \textbf{28}
  128--144.

\bibitem[{DiCiccio and Romano(2017)}]{Diciccio-Romano(2017)}
\textsc{DiCiccio, C.~J.} and \textsc{Romano, J.~P.} (2017).
\newblock {Robust Permutation Tests for Correlation and Regression
  Coefficients}.
\newblock \textit{Journal of the American Statistical Association},
  \textbf{112} 1211--1220.

\bibitem[{Dufour(2003)}]{Dufour(2003)}
\textsc{Dufour, J.-M.} (2003).
\newblock {Identification, Weak Instruments and Statistical Inference in
  Econometrics}.
\newblock \textit{Canadian Journal of Economics}, \textbf{36} 767--808.

\bibitem[{Dufour et~al.(2019)Dufour, Flachaire and
  Khalaf}]{Dufour-Flachaire-Khalaf(2019)}
\textsc{Dufour, J.-M.}, \textsc{Flachaire, E.} and \textsc{Khalaf, L.} (2019).
\newblock {Permutation Tests for Comparing Inequality Measures}.
\newblock \textit{Journal of Business \& Economic Statistics}, \textbf{37}
  457--470.

\bibitem[{Dufour and Taamouti(2007)}]{Dufour-Taamouti(2007)}
\textsc{Dufour, J.-M.} and \textsc{Taamouti, M.} (2007).
\newblock {Further Results on Projection-Based Inference in {IV} Regressions
  with Weak, Collinear or Missing Instruments}.
\newblock \textit{Journal of Econometrics}, \textbf{139} 133--153.

\bibitem[{Freedman and Lane(1983)}]{Freedman-Lane(1983)}
\textsc{Freedman, D.} and \textsc{Lane, D.} (1983).
\newblock {A Nonstochastic Interpretation of Reported Significance Levels}.
\newblock \textit{Journal of Business \& Economic Statistics}, \textbf{1}
  292--298.

\bibitem[{Ganong and J{\"a}ger(2018)}]{Ganong-Jager(2018)}
\textsc{Ganong, P.} and \textsc{J{\"a}ger, S.} (2018).
\newblock {A Permutation Test for the Regression Kink Design}.
\newblock \textit{Journal of the American Statistical Association},
  \textbf{113} 494--504.

\bibitem[{Hemerik et~al.(2020)Hemerik, Goeman and Finos}]{Hemeriketal2020}
\textsc{Hemerik, J.}, \textsc{Goeman, J.~J.} and \textsc{Finos, L.} (2020).
\newblock {Robust Testing in Generalized Linear Models by Sign Flipping Score
  Contributions}.
\newblock \textit{{Journal of the Royal Statistical Society Series B:
  Statistical Methodology}}, \textbf{82} 841--864.

\bibitem[{Hennessy et~al.(2016)Hennessy, Dasgupta, Miratrix, Pattanayak and
  Sarkar}]{Hennessyetal2016}
\textsc{Hennessy, J.}, \textsc{Dasgupta, T.}, \textsc{Miratrix, L.},
  \textsc{Pattanayak, C.} and \textsc{Sarkar, P.} (2016).
\newblock {A Conditional Randomization Test to Account for Covariate Imbalance
  in Randomized Experiments}.
\newblock \textit{{Journal of Causal Inference}}, \textbf{4} 61--80.

\bibitem[{Imbens and Rosenbaum(2005)}]{Imbens-Rosenbaum(2005)}
\textsc{Imbens, G.~W.} and \textsc{Rosenbaum, P.~R.} (2005).
\newblock {Robust, Accurate Confidence Intervals with a Weak Instrument:
  Quarter of Birth and Education}.
\newblock \textit{Journal of the Royal Statistical Society: Series A
  (Statistics in Society)}, \textbf{168} 109--126.

\bibitem[{Janssen(1997)}]{Janssen(1997)}
\textsc{Janssen, A.} (1997).
\newblock {Studentized Permutation Tests for Non-IID Hypotheses and the
  Generalized Behrens-Fisher Problem}.
\newblock \textit{Statistics \& Probability letters}, \textbf{36} 9--21.

\bibitem[{Kleibergen(2002)}]{Kleibergen(2002)}
\textsc{Kleibergen, F.} (2002).
\newblock {Pivotal Statistics for Testing Structural Parameters in Instrumental
  Variables Regression}.
\newblock \textit{Econometrica}, \textbf{70} 1781--1803.

\bibitem[{Lehmann and Romano(2005)}]{Lehmann-Romano(2005)}
\textsc{Lehmann, E.~L.} and \textsc{Romano, J.~P.} (2005).
\newblock \textit{{Testing Statistical Hypotheses}}.
\newblock 3rd ed. Springer.

\bibitem[{Li et~al.(2018)Li, Ding and Rubin}]{LiDingRubin2018}
\textsc{Li, X.}, \textsc{Ding, P.} and \textsc{Rubin, D.~B.} (2018).
\newblock {Asymptotic Theory of Rerandomization in Treatment--Control
  Experiments}.
\newblock \textit{Proceedings of the National Academy of Sciences},
  \textbf{115} 9157--9162.

\bibitem[{Linnik(2008)}]{Linnik(2008)}
\textsc{Linnik, J.~V.} (2008).
\newblock \textit{{Statistical Problems with Nuisance Parameters}}, vol.~20.
\newblock American Mathematical Soc.

\bibitem[{Mikusheva and Sun(2022)}]{Mikusheva-Sun(2022)}
\textsc{Mikusheva, A.} and \textsc{Sun, L.} (2022).
\newblock {Inference with Many Weak Instruments}.
\newblock \textit{The Review of Economic Studies}, \textbf{89} 2663--2686.

\bibitem[{Moreira and Moreira(2019)}]{Moreira-Moreira(2019)}
\textsc{Moreira, H.} and \textsc{Moreira, M.~J.} (2019).
\newblock {Optimal Two-Sided Tests for Instrumental Variables Regression with
  Heteroskedastic and Autocorrelated Errors}.
\newblock \textit{Journal of Econometrics}, \textbf{213} 398--433.

\bibitem[{Moreira(2003)}]{Moreira(2003)}
\textsc{Moreira, M.~J.} (2003).
\newblock {A Conditional Likelihood Ratio Test for Structural Models}.
\newblock \textit{Econometrica}, \textbf{71} 1027--1048.

\bibitem[{Moreira et~al.(2009)Moreira, Porter and
  Suarez}]{Moreira-Porter-Suarez(2009)}
\textsc{Moreira, M.~J.}, \textsc{Porter, J.~R.} and \textsc{Suarez, G.~A.}
  (2009).
\newblock {Bootstrap Validity for the Score Test When Instruments May be Weak}.
\newblock \textit{Journal of Econometrics}, \textbf{149} 52--64.

\bibitem[{Motoo(1956)}]{Motoo(1956)}
\textsc{Motoo, M.} (1956).
\newblock {On the Hoeffding’s Combinatorial Central Limit Theorem}.
\newblock \textit{Annals of the Institute of Statistical Mathematics},
  \textbf{8} 145--154.

\bibitem[{Neubert and Brunner(2007)}]{Neubert-Brunner(2007)}
\textsc{Neubert, K.} and \textsc{Brunner, E.} (2007).
\newblock {A Studentized Permutation Test for the Non-Parametric
  Behrens--Fisher Problem}.
\newblock \textit{Computational Statistics \& Data Analysis}, \textbf{51}
  5192--5204.

\bibitem[{Neuhaus(1993)}]{Neuhaus(1993)}
\textsc{Neuhaus, G.} (1993).
\newblock {Conditional Rank Tests for the Two-Sample Problem under Random
  Censorship}.
\newblock \textit{Annals of Statistics}, \textbf{21} 1760--1779.

\bibitem[{Pauly(2011)}]{Pauly(2011)}
\textsc{Pauly, M.} (2011).
\newblock {Discussion about the Quality of F-Ratio Resampling Tests for
  Comparing Variances}.
\newblock \textit{Test}, \textbf{20} 163--179.

\bibitem[{Pham et~al.(1989)Pham, M{\"o}cks and Sroka}]{Pham-Mocks-Sroka(1989)}
\textsc{Pham, D.~T.}, \textsc{M{\"o}cks, J.} and \textsc{Sroka, L.} (1989).
\newblock {Asymptotic Normality of Double-Indexed Linear Permutation
  Statistics}.
\newblock \textit{Annals of the Institute of Statistical Mathematics},
  \textbf{41} 415--427.

\bibitem[{Rajan and Subramanian(2008)}]{Rajan-Subramanian(2008)}
\textsc{Rajan, R.~G.} and \textsc{Subramanian, A.} (2008).
\newblock {Aid and Growth: What Does the Cross-Country Evidence Really Show?}
\newblock \textit{The Review of Economics and Statistics}, \textbf{90}
  643--665.

\bibitem[{Rosenbaum(1984)}]{Rosenbaum1984}
\textsc{Rosenbaum, P.~R.} (1984).
\newblock {Conditional Permutation Tests and the Propensity Score in
  Observational Studies}.
\newblock \textit{Journal of the American Statistical Association}, \textbf{79}
  565--574.

\bibitem[{So and Shin(1999)}]{So-Shin(1999)}
\textsc{So, B.~S.} and \textsc{Shin, D.~W.} (1999).
\newblock {Cauchy Estimators for Autoregressive Processes with Applications to
  Unit Root Tests and Confidence Intervals}.
\newblock \textit{Econometric Theory}, \textbf{15} 165--176.

\bibitem[{Stock et~al.(2002)Stock, Wright and Yogo}]{Stock-Wright-Yogo(2002)}
\textsc{Stock, J.~H.}, \textsc{Wright, J.~H.} and \textsc{Yogo, M.} (2002).
\newblock {A Survey of Weak Instruments and Weak Identification in Generalized
  Method of Moments}.
\newblock \textit{Journal of Business and Economic Statistics}, \textbf{20}
  518--529.

\bibitem[{Young(2022)}]{Young(2020)}
\textsc{Young, A.} (2022).
\newblock {Consistency without Inference: Instrumental Variables in Practical
  Application}.
\newblock \textit{European Economic Review}, \textbf{147} 104112.

\end{thebibliography}
	\newpage
	\includepdf[pages=-]{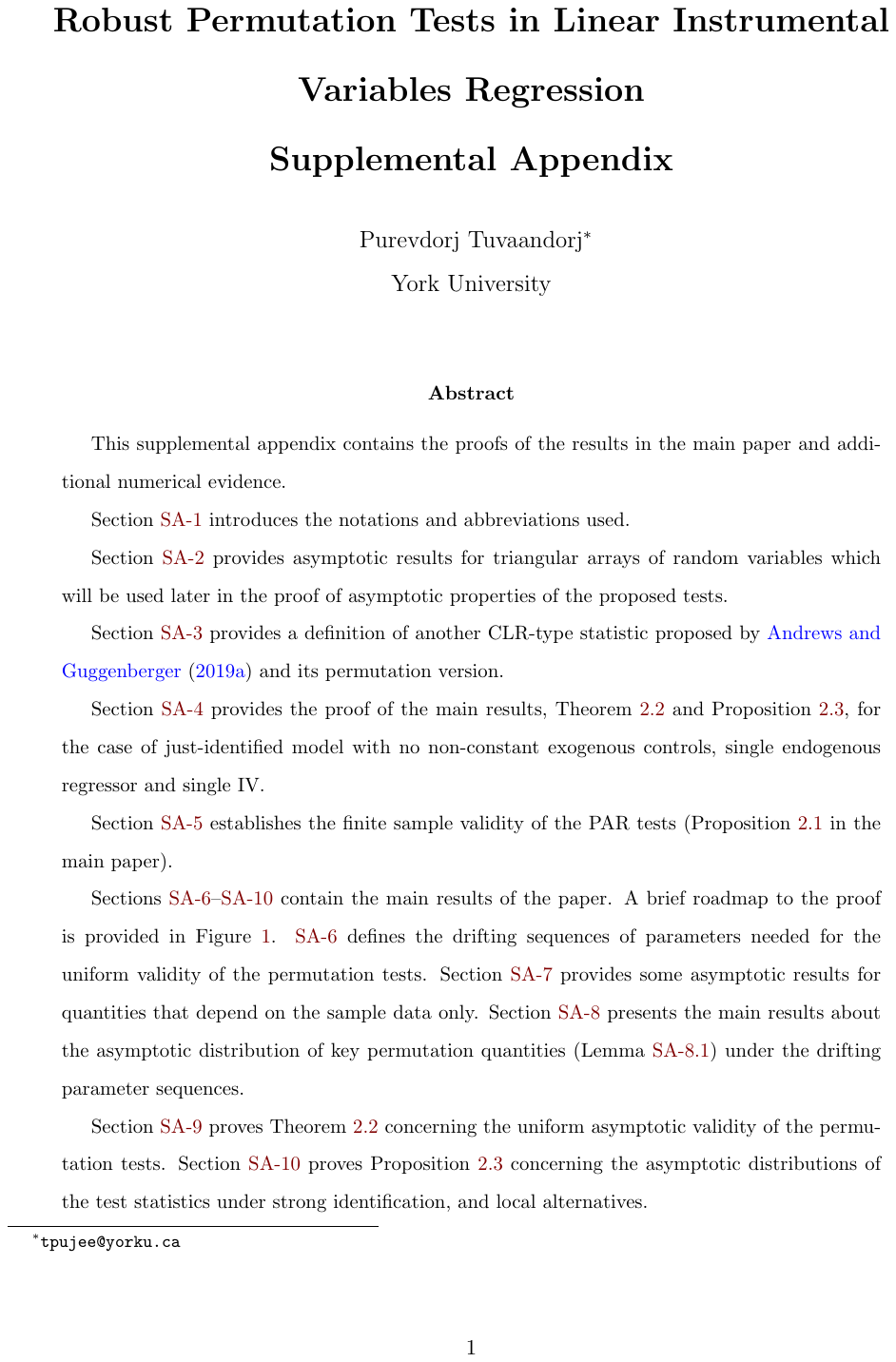}
\end{document}